\documentclass{article}

\usepackage{arxiv}

\usepackage[utf8]{inputenc} 
\usepackage[T1]{fontenc}    
\usepackage{times}
\usepackage{bm}

\usepackage{hyperref}       
\usepackage{url}            
\usepackage{booktabs}       
\usepackage{amsfonts,amssymb,amsmath}
\usepackage{nicefrac}       
\usepackage{microtype}      
\usepackage{graphicx}
\usepackage{bbm}
\usepackage{natbib}
\usepackage{color}
\usepackage{multirow}
\usepackage{caption}
\usepackage{subcaption}
\usepackage{amsthm}

\newtheorem{theorem}{Theorem}
\newtheorem{assumption}{Assumption}
\newtheorem{corollary}{Corollary}

\newcommand{\GP}{\mathcal{GP}}
\newcommand{\mc}{\mathcal}
\newcommand{\mbb}{\mathbb}

\newcommand{\Matern}{{\text{Mat\'{e}rn}}}
\newcommand{\RKHS}[1]{\mathcal{H}_{#1}}

\title{Robust and Data-Adaptive Integration of Nonconcurrent Data in Platform Trials via Gaussian Processes}

\author{
 Yuhan Qian \\
  Department of Biostatistics\\
  University of Washington\\
  Seattle, Washington, USA \\
   \And
 Yu Du \\
  Global Statistical Sciences\\
  Eli Lilly and Company\\
  Indianapolis, Indiana, USA \\
   \And
 Jingning Zhang \\
  Global Statistical Sciences\\
  Eli Lilly and Company\\
  Indianapolis, Indiana, USA \\
   \And
 Yanyao Yi \\
  Global Statistical Sciences\\
  Eli Lilly and Company\\
  Indianapolis, Indiana, USA \\
    \And
  Patrick J. Heagerty \\
  Department of Biostatistics\\
  University of Washington\\
  Seattle, Washington, USA \\
      \And
  Ting Ye \\
  Department of Biostatistics\\
  University of Washington\\
  Seattle, Washington, USA \\
}

\begin{document}
\maketitle
\begin{abstract}
A platform trial is an innovative clinical trial design that enables simultaneous and continuous evaluation of multiple treatments within a single master protocol. Existing robust methods restrict analyses to concurrently randomized participants due to concerns that including nonconcurrent data may introduce bias from temporal trends. However, this exclusion represents a missed opportunity to improve efficiency. 
We propose a Gaussian process framework for incorporating nonconcurrent data that exploits \emph{temporal smoothness}, a key feature of platform trials. The framework includes single-task and multi-task formulations and provides data-adaptive integration of nonconcurrent data with uncertainty quantification. The connection to kernel ridge regression yields a transparent frequentist interpretation of how nonconcurrent data are integrated. 
We establish two theoretical guarantees: incorporating nonconcurrent controls reduces the posterior variance of the treatment effect, and the resulting bias is controlled by a non-increasing bound.
We extend the framework to discrete outcomes and to covariate adjustment, illustrate it on a hypothetical platform trial constructed from SURMOUNT-1, and provide an implementation in the \textsf{R} package \textsf{RobinCID}.
\end{abstract}


\section{Introduction}\label{sec: intro}
A platform trial is an innovative trial design that enables simultaneous and continuous evaluation of multiple treatments within a single master protocol \citep{woodcock2017master}. Unlike traditional trials, platform trials accommodate treatments with different eligibility criteria and allow for dynamic adaptations, such as adding or dropping treatment arms \citep{berry2015platform,saville2016efficiencies, gold2022platform, burki2023platform}. 
Due to their efficiency and flexibility, platform trials have gained prominence across disease areas.
Notable examples include the Adaptive COVID-19 Treatment Trial \citep{ACTT1} in infectious diseases and the Lung-MAP \citep{herbst2015lung} and I-SPY 2 \citep{barker2009spy} trials in oncology. However, the adaptability that makes platform trials attractive also introduces complexities for statistical inference, as discussed in recent draft guidance from the Food and Drug Administration \citep{fda:2023platform}.

A fundamental step toward valid inference in platform trials is the precise definition of the target estimand and its robust estimation. \cite{qian2025estimandsrobustinferencetreatment} proposed a clinically meaningful estimand based on the \emph{entire concurrently eligible} (ECE) trial population. For comparing two treatments, the ECE trial population is a population of all individuals who meet the eligibility criteria for both treatments and could potentially be enrolled during a time period when both treatments are available, and therefore could have been randomized to either treatment. The ECE trial population extends the standard entire trial population by imposing additional restrictions on concurrency and eligibility to preserve the integrity of randomized comparisons. It is defined entirely by pre-randomization characteristics and thus remains invariant to changes in randomization ratios. Robust estimation can be achieved using weighting or stratification without any modeling assumptions, but only on concurrent data.

While this concurrent-only framework ensures validity and aligns with FDA draft guidance \citep{fda:2023platform} and recommendations from researchers \citep{lee_including_2020, dodd2021platform}, it completely excludes \emph{nonconcurrent data}, namely information from participants enrolled before one of the treatments under comparison became available, typically from the control arm and hence termed nonconcurrent controls. Appropriately incorporating such data could improve the precision and power of treatment effect estimation and testing. However, nonconcurrent participants were not randomized alongside the treatment group being compared, and their naive inclusion risks bias from temporal shifts in patient populations and outcomes. This tension between efficiency gains and potential bias has spurred substantial interest in developing principled approaches for incorporating nonconcurrent data \citep{dodd2021platform, sridhara2022use, park2022use, bofill2022commentary, bofill2023use}.

A related topic is the incorporation of external or historical data in clinical trials \citep{schmidli2020beyond, li2023improving, yi2023testing, pan2024bayesian, gao2025improving}. However, nonconcurrent data in platform trials possess distinctive features that set them apart from general external or historical data. Unlike entirely external data sources, nonconcurrent data arise from the same master protocol, share similar eligibility criteria, and are collected under a unified data infrastructure. These characteristics provide greater assurance and mitigate certain sources of bias. More importantly, because patient recruitment occurs continuously in platform trials, the distribution of enrolled patients evolves smoothly over calendar time, with patients enrolled closer together in time tending to exhibit greater distributional similarity. We call this gradual change in the enrolled population \emph{temporal smoothness}. Consequently, enrollment calendar time, often dismissed as a nuisance factor in traditional clinical trials, becomes a critical element that must be explicitly modeled in platform trials.

Several methods have been proposed to adjust for temporal trends when using nonconcurrent data in platform trials. \cite{saville_bayesian_2022} introduced the Bayesian Time Machine, which models temporal trends over discrete time intervals using a Bayesian hierarchical framework. \cite{Roig2025TreatmentControl} extended this approach from binary to continuous endpoints and showed that its performance is sensitive to prior specification, making it less robust than frequentist model-based alternatives unless the prior is calibrated. \cite{wang_bayesian_2023} proposed a temporal effect-adjusted prior that adaptively downweights nonconcurrent data when temporal changes occur between calendar time periods. Frequentist model-based methods have also been developed, adjusting for temporal trends by including enrollment time as a covariate in the regression model \citep{lee_including_2020, roig2022model}. \cite{santacatterina2025identification} developed nonparametric identification and estimation approaches for using nonconcurrent controls under strong conditional exchangeability assumptions. For a broader overview of related methods, see \cite{bofill2023use}.

Despite these advancements, existing approaches share several limitations. First, with the exception of \cite{santacatterina2025identification}, they lack formally defined causal estimands and instead focus on model parameters. Second, most rely on parametric models that may be either too simplistic to capture complex temporal dynamics or too opaque to interpret. Alternatively, nonparametric approaches require strong assumptions such as conditional exchangeability. 
Third, they give little insight into how individual nonconcurrent observations contribute to the estimate, making it hard to diagnose whether borrowing is helping or hurting in a given application.

To address these limitations, we propose a flexible and interpretable Gaussian process (GP) framework for incorporating nonconcurrent data, grounded in causal inference and machine learning. Building on the estimand defined in \cite{qian2025estimandsrobustinferencetreatment}, we target a clearly specified causal quantity that anchors the statistical analysis. We model the outcome trajectory in enrollment time using GPs, which are nonparametric Bayesian models that provide principled function estimation and uncertainty quantification \citep{rasmussen2005gaussian}.
We establish two theoretical guarantees. Incorporating nonconcurrent controls reduces posterior variance for both the outcome mean and the treatment effect function (Theorem~\ref{theorem: variance reduction}), and the resulting bias is bounded by a quantity that does not increase when nonconcurrent controls are added (Theorem~\ref{theorem: bias bound} and Corollary~\ref{cor:noninflating_bias_bound}).  We further provide a frequentist interpretation through the connection to kernel ridge regression \citep{kanagawa2025gaussianprocessesreproducingkernels}, which yields an explicit weight for each observation and quantifies the effective contribution of nonconcurrent data. The framework extends to discrete outcomes and to covariate adjustment, and we illustrate it on a hypothetical platform trial constructed from SURMOUNT-1 \citep{jastreboff2022tirzepatide}, with implementation in the \textsf{R} package \textsf{RobinCID}.

The remainder of the paper develops the GP framework and estimator in Section~\ref{sec: methods}, establishes theoretical guarantees in Section~\ref{sec: theory}, discusses extensions in Section~\ref{sec: extensions}, reports simulation results in Section~\ref{sec: simulation}, and applies the method to SURMOUNT-1 in Section~\ref{sec: real data}.

\section{Methods}\label{sec: methods}
\subsection{Setup and Estimand}
Consider a platform trial with $J$ treatment arms. Let $E\in\mathbb{R}$ denote enrollment time (calendar time, measured on a prespecified scale), $\bm X\in\mathbb{R}^p$ denote baseline covariates unaffected by treatment assignment, and $\bm W=(E,\bm X)$. Let $A$ denote the treatment assignment indicator that equals $a$ if the individual is assigned to arm $a$. Define $Y^{(a)}$ as the potential outcome 
under arm $a$, for $a=1,\dots,J$,  and let $Y=Y^{(A)}$ denote the observed outcome under the consistency assumption \citep{splawa1990application, rubin1974estimating}. We assume throughout that $(\bm W_i, Y_i^{(1)}, \dots, Y_i^{(J)}, Y_i, A_i)$, $i=1,\dots, n$, are independent and identically distributed draws from $(\bm W, Y^{(1)}, \dots, Y^{(J)}, Y, A)$ with finite second moments. 

In platform trials, randomization typically depends on design variables such as site, enrollment window, and biomarker strata, and some arms may have zero assignment probability for certain participants due to eligibility restrictions or arm unavailability. Assumption~\ref{assump: randomization} formalizes this, where $\bm Z$ is a subvector of $\bm W$ collecting the observed design variables on which randomization is based. This assumption holds by design. 

\begin{assumption}[Randomization]\label{assump: randomization}
There exists $\bm Z$ such that $A \perp ( \bm W, Y^{(a)} ) \mid \bm Z$ for $a=1,\dots, J$, and
$P(A=a\mid \bm Z)=\pi_a(\bm Z)$ for known randomization probabilities $\{\pi_a(\cdot)\}_{a=1}^J$ with $\sum_{a=1}^J \pi_a(\bm Z)=1$.
\end{assumption}

The causal estimand of interest is the contrast between expected potential outcomes under two treatments $j$ and $k$, defined within the ECE trial population  \citep{qian2025estimandsrobustinferencetreatment}:
\begin{equation}\label{eq: theta-jk}
\left( \begin{array}{c}
  \theta_{jk}   \\
    \theta_{kj}   
\end{array}\right)  
   =   \left( \begin{array}{c}
 \mathbb{E} [\, Y^{(j)} \mid \pi_j(\bm Z)>0,\pi_k(\bm Z)>0 \,] \\
  \mathbb E [\, Y^{(k)} \mid \pi_j(\bm Z)>0,\pi_k(\bm Z)>0 \,]
\end{array}\right).
\end{equation}
The ECE trial population is a population of individuals with a positive probability of being assigned to either treatment $j$ or $k$, i.e., those satisfying $\pi_j(\bm Z)>0$ and $\pi_k(\bm Z)>0$. In platform trials, it is a superpopulation represented by all trial participants who are eligible for both treatments during time periods when both treatments are actively enrolling. We focus on the linear contrast $\delta_{jk} = \theta_{jk} - \theta_{kj}$, though the framework extends to other contrasts. Although the estimand is defined within the ECE trial population, our method incorporates nonconcurrent data to improve efficiency while controlling for bias.

\subsection{Single-Task Gaussian Processes}\label{sec: single-task}
We begin by introducing the single-task GP model, which fits a separate GP per treatment arm.  In our setting, each task corresponds to a treatment. For arm $a$, the model specifies a latent function $f_a:\mathbb{R}\to\mathbb{R}$ that captures how the potential outcome $Y_i^{(a)}$ depends on enrollment time $E_i$, and an additive noise term $\epsilon_{i,a}$:
\begin{align}\label{model: uni}
    Y_i^{(a)} & =  f_a(E_i) + \epsilon_{i,a},\quad \ f_a\sim \GP(m_a,k_a),\quad \epsilon_{i,a}\sim N(0,\sigma_a^2),
\end{align}
where $\sigma_a^2$ is the noise variance (fixed or assigned a hyperprior), $\GP(m_a, k_a)$ denotes a Gaussian process with mean
function $m_a : \mathbb{R} \to \mathbb{R}$ and covariance function
(kernel) $k_a : \mathbb{R} \times \mathbb{R} \to \mathbb{R}$, and
$f_a$ and $\epsilon_{i,a}$ are mutually independent. When the inputs are vectors $\bm u=(u_1,\dots,u_p)^T\in\mbb{R}^p$ and $\bm v=(v_1,\dots,v_q)^T\in\mbb{R}^q$, we write $m_a(\bm u) = (m_a(u_1), \dots, m_a(u_p))^T\in \mathbb R^p$ and $k_a(\bm u, \bm v) = [k_a(u_r, v_s)]_{r,s}\in \mathbb R^{p\times q}$. We assume that model \eqref{model: uni} holds for all individuals who satisfy the eligibility criteria for treatment $a$, regardless of whether they are concurrent or nonconcurrent. In Section \ref{subsec: covariates}, we discuss an extension that adjusts for covariates.

The kernel $k_a$ encodes temporal smoothness in $f_a$ and determines how information is shared across enrollment times. For example, the squared exponential kernel
\begin{align}
    k_{\rm SE}(e,e';\alpha,h)=
    \alpha^2\exp\left\{-\frac{|e-e'|^2}{2h^2}\right\} \label{kernel: SE}
\end{align}
has variance parameter $\alpha^2$ governing the marginal variability of
$f_a$ and length-scale $h$ controlling the range of temporal dependence. A larger $h$ enables greater borrowing from nonconcurrent data, while a smaller $h$ limits dependence to nearby enrollment times. 

We now derive the posterior distribution of $f_a$, which is also a GP and forms the basis for Bayesian inference. 
Let $n_a$ denote the number of individuals with $A_i=a$, and let $\bm{Y}_a$ and $\bm{E}_a$ denote the corresponding outcome and enrollment time vectors.
The posterior distribution of $f_a$ is
\begin{equation}\label{eq: posterior GP}
f_a \mid \bm{Y}_a, \bm{E}_a \sim \mathcal{GP}(\bar{m}_a, \bar{k}_a),
\end{equation}
where the posterior mean function $\bar{m}_a$ and posterior covariance function $\bar{k}_a$ are
\begin{align}
    \bar{m}_a(e) &= m_a(e) + k_a(e,\bm E_a) \left\{k_a(\bm E_a,\bm E_a)+\sigma_a^2\bm I_{n_a}\right\}^{-1} \left\{\bm{Y}_a - m_a(\bm{E}_a)\right\}, \label{eq: single posterior mean}\\
    \bar{k}_a(e, e') &= k_a(e, e') - k_a(e,\bm E_a) \left\{k_a(\bm E_a,\bm E_a)+\sigma_a^2\bm I_{n_a}\right\}^{-1} k_a(\bm E_a,e), \nonumber
\end{align}
and $\bm{I}_{n}$ denotes the $n \times n$ identity matrix. It is common to set the prior mean to $m_a \equiv 0$ \citep{rasmussen2005gaussian}. The posterior mean~\eqref{eq: single posterior mean} reveals a data-adaptive borrowing mechanism: the estimated outcome mean at enrollment time $e$ is a weighted combination of all observed outcomes $\bm{Y}_a$, with weights $k_a(e,\bm E_a) \left\{k_a(\bm E_a,\bm E_a)+\sigma_a^2\bm I_{n_a}\right\}^{-1}$ that depend on the length-scale $h$ and the proximity of $e$ to the observed enrollment times. Consequently, nonconcurrent participants contribute more to the estimate at a concurrent time point when the temporal trend is smooth and their enrollment times are close to the concurrent period.

\subsection{Multi-Task Gaussian Processes}\label{sec:multi-task Gaussian Processes}

Single-task GP models fit each arm's outcome trajectory independently, without imposing any shared structure across arms. However, in platform trials, temporal trends are often largely shared across arms, driven for example by shifts in baseline risk or the patient population, while treatment effects remain comparatively stable. As a result, the difference between treatment and control outcomes may be smoother over time than either arm-specific trajectory alone.

To exploit this structure, we adopt a multi-task GP framework that decomposes potential outcomes into a shared baseline function $f$ and a treatment effect function $\Delta$. Specifically,
for all individuals eligible for both treatments $j$ and $k$, whether concurrent or nonconcurrent, we assume:
\begin{align*}
  Y^{(j)}_i = f(E_i) + \Delta(E_i) + \epsilon_{i,j}, \quad  Y^{(k)}_i = f(E_i) + \epsilon_{i,k},
\end{align*}
where $f$ and $\Delta$ follow independent GP priors with kernel functions $k_f$ and $k_\Delta$ and mean functions $m$ and $\delta$, respectively, and $\epsilon_{i,a} \sim N(0, \sigma_a^2)$ for $a \in \{j, k\}$. Here, the prior mean function of $\Delta$ is a constant $\delta \in \mathbb{R}$, representing the prior guess for $\delta_{jk}$, and $\Delta$ captures enrollment-time-dependent deviations from this constant. This is a special case of the linear model of coregionalization \citep{alvarez_kernels_2012}, and can be written in compact notation as
\begin{align}\label{model: multi-task GP}
    \begin{pmatrix}
        Y_i^{(j)}\\ Y_i^{(k)}
    \end{pmatrix} &= \bm{f}_{jk}(E_i) + \bm{\epsilon}_i, \quad \bm{f}_{jk} \sim \GP(\bm{m}, \bm{k}), \quad \bm{\epsilon}_i \sim N\left(\mathbf{0}_2, \begin{pmatrix}
        \sigma_j^2 & 0\\ 0 & \sigma_k^2
    \end{pmatrix}\right),\\
    \bm{m} &= \begin{pmatrix}
        m+\delta \\ m
    \end{pmatrix}, \quad 
    \bm{k}(e,e') = \begin{pmatrix}
        1 & 1 \\
        1 & 1
    \end{pmatrix} \otimes k_f(e,e') + \begin{pmatrix}
        1 & 0 \\
        0 & 0
    \end{pmatrix} \otimes k_\Delta(e,e'), \nonumber
\end{align}
where $\bm{m}$ is a vector-valued mean function, $\bm{k}$ is a $2 \times 2$ matrix-valued covariance function capturing both within-arm variance and between-arm covariance, $\mathbf{0}_n$ is an $n$-dimensional column vector of zeros, and $\otimes$ denotes the Kronecker product. 

Although our multi-task GP model imposes a correlation structure between arms, it retains flexibility through kernel hyperparameters. Specifically, the variance of $k_\Delta$ controls the degree of treatment effect variation over time: larger values accommodate substantial temporal heterogeneity, while smaller values constrain $\Delta$ to remain nearly constant. Because this parameter is learned from data, the model adapts to the observed degree of treatment effect stability.

For estimation and uncertainty quantification of the treatment effect, we derive the posterior distribution of $\Delta$. Let $\bm E = (\bm E_j^T, \bm E_k^T)^T$.  The posterior distribution
of the treatment effect function $\Delta$ evaluated at a vector of
enrollment times $\bm e$ is $\Delta(\bm e) \mid \bm{Y}_j, \bm{Y}_k, \bm{E} \sim N(\bar{m}_\Delta(\bm e), \bar{k}_\Delta(\bm e, \bm e))$, where
\begin{align}\label{eq: multi-GP posterior mean and covariance}
    \bar{m}_\Delta(\bm e) &= \delta + k_{\Delta}(\bm e, \bm{E}_j) B_1^{-1} \left\{(\bm{Y}_j - m(\bm{E}_j)-\delta\bm 1_{n_j}) + B_2 (\bm{Y}_{k} - m(\bm{E}_k))\right\},\\
    \bar{k}_\Delta(\bm e,\bm e) &= k_{\Delta}(\bm e, \bm e) - k_{\Delta}(\bm e, \bm{E}_j) B_1^{-1} k_{\Delta}(\bm{E}_j, \bm e), \nonumber
\end{align}
with $\bm 1_n$ denoting a length-$n$ vector of ones, 
\(
    B_1 = k_{f,jj} + k_{\Delta,jj} + \sigma_j^2 \bm{I}_{n_j} - k_{f,jk} (k_{f,kk} + \sigma_{k}^2 \bm{I}_{n_k})^{-1} k_{f,kj}\) and \(
    B_2 = -k_{f,jk} (k_{f,kk} + \sigma_k^2 \bm{I}_{n_k})^{-1}.\)
Here, $k_{f,jk} := k_f(\bm{E}_j, \bm{E}_k)$, and $k_{\Delta,jj}$, $k_{f,jj}$, $k_{f,kk}$, and $k_{f,kj}$ are defined analogously.

\subsection{Estimation and Inference}\label{sec: estimation}

We describe estimation and inference for the treatment effect $\delta_{jk} = \theta_{jk} - \theta_{kj}$, defined in~\eqref{eq: theta-jk} as expectations over the ECE trial population. 
Our estimator fits the GP (single-task or multi-task) using all eligible participants, both concurrent and nonconcurrent, but averages over the concurrently eligible only, so that the estimand remains anchored to the ECE trial population while nonconcurrent data adaptively inform the GP fit.

For the single-task GP model, we fit the model using both concurrent and nonconcurrent data and obtain the posterior distributions of $f_j$ and $f_k$ as in~\eqref{eq: posterior GP}. We then estimate $\delta_{jk}$ via Bayesian bootstrap \citep{rubin1981bayesianbootstrap, oganisian_practical_2021}, which accounts for both posterior uncertainty in the outcome functions and sampling variability from the ECE trial population. Let $\mathcal{I}_{jk} = \{i: \pi_j(\bm{Z}_i) > 0, \pi_k(\bm{Z}_i) > 0\}$ denote the index set of concurrently eligible individuals and $n_{jk} = |\mathcal{I}_{jk}|$ its cardinality. For each treatment arm $a \in \{j, k\}$, we draw $B$ posterior samples $\{\hat{f}_a^{(b)}(\bm e)\}_{b=1}^B$ from the posterior distribution of $f_a(\bm e)$ in~\eqref{eq: posterior GP}, where $\bm e$ collects the enrollment times $\{e_i : i \in \mathcal{I}_{jk}\}$. For each $b = 1, \dots, B$, a bootstrap replicate of the treatment effect estimate is constructed as follows: given posterior samples $\hat{f}_j^{(b)}(e_i)$ and $\hat{f}_k^{(b)}(e_i)$ for all $i \in \mathcal{I}_{jk}$, draw a weight vector $(p_i^{(b)})_{i\in \mathcal{I}_{jk}}$ from a Dirichlet distribution $\text{Dir}(1, \dots, 1)$, and compute
\begin{equation*}
\hat{\delta}_{jk}^{(b)} = \sum_{i \in \mathcal{I}_{jk}} p_i^{(b)} \{\hat{f}_j^{(b)}(e_i) - \hat{f}_k^{(b)}(e_i)\}.
\end{equation*}
The empirical distribution of $\{\hat{\delta}_{jk}^{(b)}\}_{b=1}^B$ approximates the posterior distribution of $\delta_{jk}$.

For the multi-task GP model, the procedure is analogous, with $\hat{f}_j^{(b)}(e_i) - \hat{f}_k^{(b)}(e_i)$ replaced by posterior draws of $\Delta(e_i)$ from the closed-form posterior in~\eqref{eq: multi-GP posterior mean and covariance}. 

\subsection{Kernel}\label{sec: kernel}

The covariance function, or kernel, encodes prior beliefs about temporal smoothness and is central to a GP model.
Two commonly used choices are the squared exponential kernel~\eqref{kernel: SE}, which corresponds to infinitely smooth temporal trends, and the Mat\'ern kernel \citep{rasmussen2005gaussian}, which allows explicit control over smoothness through a hyperparameter $\nu$:
\begin{align*}
    k_{\text{Mat}}(e, e'; \alpha, \nu, h) = \alpha^2 \frac{2^{1-\nu}}{\Gamma(\nu)} \left( \sqrt{2\nu} \frac{|e - e'|}{h} \right)^\nu K_\nu \left( \sqrt{2\nu} \frac{|e - e'|}{h} \right),
\end{align*}
where $\Gamma$ is the gamma function and $K_\nu$ is the modified Bessel function of the second kind \citep{porcu2023maternmodeljourneystatistics}.
The length scale $h$ controls how rapidly correlation decays with the time gap $|e-e'|$, and the parameter $\nu$ controls how gradually the temporal trend changes between nearby times. 
Larger values of either parameter increase borrowing across enrollment times.
The special case $\nu = 1/2$ yields the Laplace kernel $\alpha^2 \exp\{-|e - e'|/h\}$, and $\nu \to \infty$ recovers the squared exponential kernel \citep{kanagawa_gaussian_2018}.

\section{Frequentist Interpretation and Theoretical Guarantees}\label{sec: theory}

\subsection{Connection with Kernel Ridge Regression}\label{sec: connection with RKHS}
GPs are closely connected with kernel ridge regression \citep{kanagawa_gaussian_2018,rasmussen2005gaussian}. Let $k$ be a \emph{positive definite kernel}, that is, a symmetric function such that for any finite-dimensional vector $\bm u$, the matrix $k(\bm u,\bm u)$ is positive semidefinite. The squared exponential and {\Matern} kernels described in Section \ref{sec: kernel} are canonical examples. The reproducing kernel Hilbert space (RKHS) $\mathcal{H}_k$ induced by $k$ is a Hilbert space of real-valued functions on $\mathbb{R}$, equipped with inner product $\langle\cdot,\cdot\rangle_{\mc H_k}$ and norm $||f||_{\RKHS{k}}=\langle f,f\rangle_{\RKHS{k}}^{1/2}$.

For the single-task model~\eqref{model: uni}, the kernel ridge regression estimator minimizes the regularized empirical risk \citep{kanagawa_gaussian_2018}:
\begin{align*}
    \hat f_a = \arg\min_{f\in\mc H_{k_a}}\frac{1}{n_a}\sum_{i: A_i= a}\{Y_i - f(E_i)\}^2+\lambda\|f\|_{\mc H_{k_a}}^2,
\end{align*}
with solution $
    \hat f_a(e)=k_a(e,\bm E_a)\left\{
    k_a(\bm E_a,\bm E_a)+n_a\lambda \bm I_{n_a}\right\}^{-1}\bm Y_a, $
where $\lambda>0$ is a regularization parameter controlling the smoothness of the estimator. Setting $m_a \equiv 0$ and $\sigma_a^2 = n_a \lambda$ in the GP posterior
mean~\eqref{eq: single posterior mean}
recovers this estimator exactly, establishing a well-known equivalence between kernel ridge regression and GP inference.

We now extend this connection to the multi-task setting. For ease of presentation, set $m \equiv 0$ and rewrite the multi-task model~\eqref{model: multi-task GP} by incorporating the treatment assignment into the kernel:
\begin{align*}
Y_i &= \delta\mathbb I(A_i=j) + u(E_i,A_i) + \eta_i,\quad u \sim \GP(0,k_u),\quad 
\eta_i \sim N(0,\sigma_{A_i}^2), 
\end{align*}
where $k_u$ is a composite kernel on $\mathbb R\times \{j,k\}$ defined by $ k_u((e,a),(e',a'))=k_f(e,e')+k_t(a,a')\cdot k_\Delta(e,e')$ with $k_t(a,a')=\mbb I(a=a'=j)$. Since $k_f$, $k_t$, and $k_\Delta$ are all positive definite, and $k_u$ is the sum of $k_f$ and the product of $k_t$ and $k_\Delta$, it inherits positive definiteness \citep{aronszajn1950theory}. 

In the frequentist view, the posterior mean of $u$ coincides with the minimizer of a regularized empirical risk in the RKHS $\RKHS{k_u}$ associated with $k_u$:
\begin{align}\nonumber
    \hat u = \arg\min_{u\in \RKHS{k_u}} \frac{1}{n_j+n_k}\sum_{i: A_i\in \{j, k\}} \left\{Y_i - \delta\mathbb I(A_i=j) - u(E_i,A_i)\right\}^2/\sigma_{A_i}^2 + \lambda \|u\|^2_{\RKHS{k_u}},
\end{align}
where $\delta$ is treated as a known constant. The solution is
\begin{align}\nonumber
    \hat u(e,a)=\delta \mathbb{I}(a=j) + k_u((e,a),(\bm E,\bm A))\{k_u((\bm E,\bm A),(\bm E,\bm A))+(n_j+n_k)\lambda \Sigma\}^{-1}\{\bm Y-(\delta\bm 1_{n_j}^T, \mathbf{0}_{n_k}^T)^T\},
\end{align}
where $\Sigma = \textup{diag}(\sigma_{j}^2 \bm{1}_{n_j}^T, \sigma_{k}^2 \bm{1}_{n_k}^T)$, $\bm E = (\bm E_j^T, \bm E_k^T)^T$, $\bm Y = (\bm Y_j^T, \bm Y_k^T)^T$, and $\bm A = (j\bm{1}_{n_j}^T, k\bm{1}_{n_k}^T)^T$.
With $(n_j+n_k)\lambda = 1$, the difference $\hat u(e, j) - \hat u(e, k)$ equals the posterior mean of $\Delta(e)$ in \eqref{eq: multi-GP posterior mean and covariance}.

\subsection{Weights and Effective Sample Size}
As shown in Section \ref{sec: connection with RKHS},  the posterior mean of both the single-task~\eqref{eq: single posterior mean} and multi-task~\eqref{eq: multi-GP posterior mean and covariance} models is a weighted linear combination of observed outcomes. For each concurrently eligible individual $i$, the posterior mean of the treatment effect can be written as $\bm w(e_i)^T \bm y + c$, where $\bm w:\mathbb{R}\to \mbb{R}^{n_j+n_k}$ is a weight function, $\bm y$ is a transformed outcome vector, and $c$ is a constant determined by the prior means.  For the single-task model, the weight function is $\bm w(\cdot) = (\bm w_j(\cdot)^T, -\bm w_k(\cdot)^T)^T$ and the transformed outcome vector is $\bm y = (\bm Y_j^T - m_j(\bm E_j)^T, \bm Y_k^T - m_k(\bm E_k)^T)^T$, where 
$\bm w_a(\cdot) = \{k_a(\bm E_a, \bm E_a) + \sigma_a^2 \bm I_{n_a}\}^{-1}
k_a(\bm E_a, \cdot)$ is the weight function for arm~$a$. 
For the multi-task model, the weight function is $\bm w(\cdot)=(\bm I_{n_j},B_2)^TB_1^{-1}k_{\Delta}(\bm{E}_j,\cdot)$ and the transformed outcome vector is $(\bm{Y}_j^T - m(\bm{E}_j)^T-\delta\bm 1_{n_j}^T, \bm{Y}_{k}^T - m(\bm{E}_k)^T)^T$. Collecting these weight vectors across all concurrently eligible individuals gives the weight matrix $W_{cc}=(\bm w(e_1),\bm w(e_2),\dots,\bm w(e_{n_{jk}}))^T\in\mathbb{R}^{n_{jk}\times (n_j+n_k)}$, so that the estimator of $\delta_{jk}$ is $\mathbf{1}_{n_{jk}}^T W_{cc} {\bm y} /n_{jk}$. This makes explicit how each concurrent and nonconcurrent outcome contributes to the estimate. 

Let $(w_1, w_2, \dots, w_{n_k})^T$ denote the weights for the control arm outcomes in estimating $\delta_{jk}$. Then the effective sample size for the control arm is 
$\mathrm{ESS} = (\sum_i w_i)^2 / \sum_i w_i^2$ \citep{kish1965survey}, representing the size of an unweighted control sample that would yield the same precision. A larger ESS indicates greater variance reduction from incorporating nonconcurrent controls. Since the GP weights $w_i$ can be negative while the ESS formula assumes nonnegative weights, we report ESS as an approximate diagnostic only.

\subsection{Theoretical Guarantees for Variance and Bias}

In this section, we show that incorporating nonconcurrent data reduces the posterior variance of the outcome mean under the single-task GP model~\eqref{model: uni} and of the treatment effect function $\Delta$ under the multi-task GP model~\eqref{model: multi-task GP}, without inflating the bias bound.

\begin{theorem}[Variance Reduction]\label{theorem: variance reduction}
Consider the GP models with fixed kernel hyperparameters. For $* \in \{a, \Delta\}$, let $\bar{k}_{*,\textup{cc}}$ denote the posterior covariance function when fitting the model using only concurrent data, and let $\bar{k}_{*}$ denote the posterior covariance function when incorporating nonconcurrent control data. Then for any vector of enrollment times $\bm{e}$:
\begin{enumerate}
    \item[(i)] Under the single-task GP model~\eqref{model: uni}, the posterior covariance function of $f_a$ satisfies $ \bar{k}_{a,\textup{cc}}(\bm{e}, \bm{e}) - \bar{k}_{a}(\bm{e}, \bm{e}) \succeq 0.$  Here, $\succeq 0$ denotes positive semi-definiteness.
    \item[(ii)] Under the multi-task GP model~\eqref{model: multi-task GP}, the posterior covariance function of the treatment effect $\Delta$ satisfies $
        \bar{k}_{\Delta,\textup{cc}}(\bm{e}, \bm{e}) - \bar{k}_{\Delta}(\bm{e}, \bm{e}) \succeq 0.$
\end{enumerate}
\end{theorem}

Theorem~\ref{theorem: variance reduction}(i) establishes that incorporating nonconcurrent data reduces posterior uncertainty in the outcome model for treatment arm $a$ under the single-task GP. Part~(ii) extends this result to the treatment effect function $\Delta$ under the multi-task GP. Consequently, Theorem~\ref{theorem: variance reduction} implies variance reduction for any weighted average estimator with positive weights, including the arm-specific outcome mean estimator under the single-task GP and the treatment effect estimator under the multi-task GP via the Bayesian bootstrap of Section~\ref{sec: estimation}. 

These results hold conditional on the kernel hyperparameters, comparing the concurrent-only and nonconcurrent-incorporated estimators at the same fixed hyperparameter values. In practice, hyperparameters are learned from data via marginal likelihood maximization or MCMC \citep{rasmussen2005gaussian}, and may be estimated differently by the two estimators. The theoretical guarantees therefore do not apply unconditionally. Nonetheless, the simulation results in Section~\ref{sec: simulation} demonstrate that the efficiency gains persist under data-driven hyperparameter estimation.

Theorem~\ref{theorem: variance reduction} follows from a projection argument. Under the single-task GP, the posterior covariance is
$\bar{k}(\bm{e},\bm{e}) = k(\bm{e},\bm{e}) - k(\bm{e},\bm{E})\{k(\bm{E},\bm{E})+\sigma^2\bm{I}\}^{-1}k(\bm{E},\bm{e})$. Under the multi-task model derived in \eqref{eq: multi-GP posterior mean and covariance}, $B_1$ replaces $k(\bm{E},\bm{E})+\sigma^2\bm{I}$. 
The subtracted term is the variance explained by projecting onto the span of the observed data in the reproducing kernel Hilbert space. Incorporating nonconcurrent controls enlarges $\bm{E}$ and thus the subspace. Since orthogonal projections onto nested subspaces are ordered in the positive semidefinite sense, the posterior covariance is non-increasing, reflecting the familiar fact that conditioning a Gaussian on additional variables never increases conditional variance.

For the treatment effect $\Delta$ under the multi-task model, the argument is less direct. Because $\Delta$ appears only in the treatment arm, control-arm data cannot inform it directly. Instead, they improve estimation of the shared baseline $f$. Express $B_1$ as $B_1 = k_{\Delta,jj} + \sigma_j^2\bm{I}_{n_j}+ F_{\textup{Sch}}$, where $F_{\textup{Sch}} = k_{f,jj} -k_{f,jk}(k_{f,kk}+\sigma_k^2\bm{I}_{n_k})^{-1}k_{f,kj}$ is the conditional variance of $f(\bm{E}_j)$ given the control data. Adding nonconcurrent controls reduces $F_{\textup{Sch}}$ and hence $B_1$. Since $\bar{k}_\Delta = k_\Delta - k_\Delta\, B_1^{-1}\, k_\Delta^T$, a smaller $B_1$ means a larger subtracted term and thus lower posterior variance. Intuitively, richer control data improve estimation of $f$, which in turn yields a more precise estimate of $\Delta$.

Theorem~\ref{theorem: variance reduction} also has a direct estimation-error interpretation. The posterior variance $\bar{k}(e,e)$ equals $\mathbb{E}_{f\sim\GP(\bar m,\bar {k})}[\{f(e)-\bar m(e)\}^2]$, the average squared estimation error at enrollment time $e$ under the GP posterior. The theorem therefore establishes that incorporating nonconcurrent controls reduces this average squared estimation error.

Beyond variance reduction, a natural concern is whether nonconcurrent controls introduce bias in the presence of temporal trends. We next examine bias from using nonconcurrent controls under a fixed data-generating process, treating the GP as a working model. 
\begin{assumption}[Data-generating process]\label{assump:dgp}
     Assume the structural model $Y^{(a)}=f_0(E)+\mbb I(a=j)\Delta_0(E)+\epsilon_a$ for $a\in\{j,k\}$, where $f_0,\Delta_0$ are deterministic functions and $\epsilon_a$ is independent of $(A,E)$ with mean zero and variance $\sigma_a^2.$ 
\end{assumption}

We work conditionally on the enrollment times $(\bm E_j, \bm E_k)$ and the concurrently eligible enrollment times $(E_i)_{i\in\mathcal{I}_{jk}}$, and treat all kernel hyperparameters as fixed.
Let $F$ denote the marginal distribution of $E$ in the ECE trial population, so that $\delta_{jk}= \mathbb{E}_F[\Delta_0(E)]$ by Assumption~\ref{assump:dgp} and iterated expectations. Let $F_n = n_{jk}^{-1}\sum_{i\in\mathcal{I}_{jk}} \delta_{E_i}$ denote the empirical distribution of the concurrently eligible enrollment times, and define the sample-level estimand
$\tilde\delta_{jk} = \mathbb{E}_{F_n}[\Delta_0(E)] = n_{jk}^{-1}\sum_{i\in\mathcal{I}_{jk}} \Delta_0(E_i)$.
Our estimator is $\hat\psi := B^{-1}\sum_{b=1}^B \hat\delta_{jk}^{(b)}$, the average of the Bayesian bootstrap draws from Section~\ref{sec: estimation}. Define the kernel mean embedding under $F_n$: $\mu_n(\cdot) = \mathbb{E}_{F_n}[k_\Delta(E,\cdot)] = n_{jk}^{-1}\sum_{i\in\mathcal{I}_{jk}} k_\Delta(E_i,\cdot)$.

\begin{theorem}[Bias bound]\label{theorem: bias bound}
Assume $\Delta_0 - \delta \in \mathcal{H}_{k_\Delta}$ and $f_0 - m \in \mathcal{H}_{k_f}$. Then, conditional on $(\bm E_j,\bm E_k,(E_i)_{i\in\mathcal{I}_{jk}})$,
\begin{equation}\label{eq: bias bound thm}
    \left|\mathbb{E}[\hat\psi \mid \bm E_j, \bm E_k, (E_i)_{i\in\mathcal{I}_{jk}}] - \tilde\delta_{jk}\right| \leq \|\Delta_0 - \delta\|_{\mathcal{H}_{k_\Delta}}\, \left[\mathbb{E}_{F_n}\{\bar{k}_\Delta(E,E)\}\right]^{\frac{1}{2}}
    \;+\; \|f_0 - m\|_{\mathcal{H}_{k_f}}\, \left\{{\mu_n(\bm E_j)\, B_1^{-1}\, G_f\, B_1^{-1}\, \mu_n(\bm E_j)^T}\right\}^{\frac{1}{2}},
\end{equation}
where $\mathbb{E}_{F_n}\{\bar{k}_\Delta(E,E)\} = n_{jk}^{-1}\sum_{i\in\mathcal{I}_{jk}}\bar{k}_\Delta(E_i,E_i)$ is the average posterior variance of $\Delta$ over the concurrently eligible individuals, and $G_f$ is the $n_j \times n_j$ positive semidefinite Gram matrix
\begin{equation}\label{eq: Gram matrix}
    G_f = F_{\textup{Sch}} - \sigma_k^2\, k_{f,jk}\, (k_{f,kk} + \sigma_k^2 \bm I_{n_k})^{-2}\, k_{f,kj} \preceq F_{\textup{Sch}}.
\end{equation}
\end{theorem}
The bound~\eqref{eq: bias bound thm} measures the bias of $\hat\psi$ relative to the sample-level estimand $\tilde\delta_{jk}$, which equals the average of the true treatment effect function $\Delta_0$ over the observed concurrently eligible enrollment times. By the triangle inequality, the bias relative to the population estimand $\delta_{jk}$ satisfies
\begin{equation}\label{eq:bias-decomposition}
    \left|\mathbb{E}[\hat\psi \mid \bm E_j, \bm E_k, (E_i)_{i\in\mathcal{I}_{jk}}] - \delta_{jk}\right|
    \leq
    \left|\mathbb{E}[\hat\psi \mid \bm E_j, \bm E_k, (E_i)_{i\in\mathcal{I}_{jk}}] - \tilde\delta_{jk}\right|
    + |\tilde\delta_{jk} - \delta_{jk}|,
\end{equation}
where the first term is bounded by~\eqref{eq: bias bound thm} and the second satisfies $|\tilde\delta_{jk} - \delta_{jk}| = O_p(n_{jk}^{-1/2})$ by the central limit theorem, since the $E_i$ are independent draws from $F$ with $\textup{Var}_F(\Delta_0(E)) < \infty$.

In~\eqref{eq: bias bound thm}, the first term bounds the treatment-effect bias and is governed by $\mathbb{E}_{F_n}\{\bar{k}_\Delta(E,E)\}$, which by Theorem~\ref{theorem: variance reduction}(ii) does not increase when nonconcurrent controls are incorporated. The second term captures the baseline-misspecification contribution. It vanishes when $m = f_0$ and otherwise depends on $G_f$, which measures residual uncertainty in the shared baseline $f$ after conditioning on control data. This term is not monotone in general, since $F_{\textup{Sch}}$ decreases while $B_1^{-1}$ increases as nonconcurrent controls are added. The next corollary gives a coarser but monotone bound.

\begin{corollary}[A non-increasing bias bound]\label{cor:noninflating_bias_bound}
Under the assumptions of Theorem~\ref{theorem: bias bound},
\begin{equation}\label{eq: bias bound cor}
    \left|\mathbb{E}[\hat\psi \mid \bm E_j, \bm E_k,(E_i)_{i\in\mathcal{I}_{jk}}] - \tilde\delta_{jk}\right|
\le
\Bigl(\|\Delta_0-\delta\|_{\mathcal H_{k_\Delta}}+\|f_0-m\|_{\mathcal H_{k_f}}\Bigr)
\left[\mathbb E_{F_n}\{\bar k_\Delta(E,E)\}\right]^{\frac{1}{2}}.
\end{equation}
Moreover, the right-hand side does not increase when nonconcurrent controls are incorporated.
\end{corollary}

The corollary pools both misspecification terms under a single factor $[\mathbb{E}_{F_n}\{\bar{k}_\Delta(E,E)\}]^{1/2}$, which does not increase when nonconcurrent controls are incorporated, following from Theorem \ref{theorem: variance reduction}(ii). This bound is most informative when baseline misspecification is small relative to treatment-effect misspecification, in which case it is nearly sharp and the monotonicity guarantee is meaningful. When baseline misspecification dominates, Theorem~\ref{theorem: bias bound} provides a more refined decomposition.

Together, Theorems~\ref{theorem: variance reduction} and~\ref{theorem: bias bound} and Corollary~\ref{cor:noninflating_bias_bound} establish a favorable bias-variance profile for incorporating nonconcurrent controls: posterior variance is guaranteed to decrease, and the bias, though not guaranteed to decrease, is controlled by a non-increasing bound, so the GP-based approach improves estimation efficiency without unduly inflating bias when the working model is reasonable.

\section{Extensions}\label{sec: extensions}
\subsection{Generalized GP for Discrete Outcomes}

We extend the framework to accommodate discrete outcomes such as binary or count responses. Analogous to generalized linear models \citep{Nelder1972GLM}, we replace the Gaussian observation model with an exponential family likelihood and introduce a link function $g(\cdot)$ connecting the conditional mean to the latent GP. In the single-task setting, the potential outcome under arm $a$ is modeled as
\(
    Y_i^{(a)} \mid f_a(E_i) \sim p\bigl(y \mid g^{-1}(f_a(E_i))\bigr), \ f_a \sim \GP(m_a, k_a),
\)
where $g^{-1}$ is the inverse link function and $g^{-1}(f_a(E_i)) = \mathbb{E}[Y_i^{(a)} \mid f_a(E_i)]$ is the conditional mean. For binary outcomes with the logit link, this gives $Y_i^{(a)} \mid f_a(E_i) \sim \mathrm{Bernoulli}\bigl((1+\exp\{-f_a(E_i)\})^{-1}\bigr)$, so that $f_a$ operates on the log-odds scale. For count outcomes with the log link, it gives $Y_i^{(a)} \mid f_a(E_i) \sim \mathrm{Poisson}(\exp\{f_a(E_i)\})$, so that $f_a$ operates on the log-rate scale. In both cases, the GP prior on $f_a$ enables the temporal information borrowing described in Section~\ref{sec: single-task}. The multi-task extension follows analogously by replacing the Gaussian likelihood in~\eqref{model: multi-task GP} with componentwise exponential family likelihoods applied to the latent vector $\bm{f}_{jk}(E_i)$.

Since the posterior distribution is no longer analytically tractable under a non-Gaussian likelihood, inference proceeds via approximation methods such as the Laplace approximation \citep{WilliamsBayesianClassification1998}, expectation propagation \citep{Minka2001ApproximateBayes}, or Markov chain Monte Carlo \citep{Neal1999RegressionandClassification, rasmussen2005gaussian}. Once approximate posterior samples are obtained, estimation of treatment contrasts follows the procedure in Section~\ref{sec: estimation}.

\subsection{Covariate Adjustment} \label{subsec: covariates}
The models in Sections~\ref{sec: single-task} and~\ref{sec:multi-task Gaussian Processes} model the outcome as a function of enrollment time $E$ plus independent noise. When outcomes also depend on baseline covariates $\bm X$ that vary with enrollment time, these covariates are absorbed into the residual $\epsilon_{i,a}$, violating the independent noise assumption in~\eqref{model: uni} and~\eqref{model: multi-task GP}. Adjusting for $\bm X$ is therefore important to ensure model validity and to reduce residual variance, thereby improving efficiency.

A direct approach is to include covariates in the GP mean function. Let $\bm X_i \in \mathbb{R}^p$ denote the pre-treatment covariate vector for individual $i$. For the single-task model~\eqref{model: uni}, linear covariate adjustment gives
\(
    Y_i^{(a)} = \bm X_i^T \beta_a + f_a(E_i) + \epsilon_{i,a},
\)
where $\beta_a \in \mathbb{R}^p$ captures linear covariate effects and 
$\epsilon_{i,a}\mid (E_i,\bm X_i)\stackrel{\mathrm{i.i.d.}}{\sim}N(0,\sigma_a^2)$.
The coefficient $\beta_a$ can be assigned a prior and marginalized out, yielding a GP with an augmented covariance function
\citep{rasmussen2005gaussian}. For the multi-task model~\eqref{model: multi-task GP}, covariate adjustment enters analogously through the mean functions of the shared baseline process $f$ and, optionally, through the treatment effect function $\Delta$ to capture treatment effect heterogeneity with respect to $\bm X$. In both cases, the posterior distributions retain the same structural form as in~\eqref{eq: single posterior mean} and~\eqref{eq: multi-GP posterior mean and covariance}.

The mean-function approach above captures linear covariate effects. To accommodate nonlinear effects, one can additionally incorporate covariates through the kernel. A natural choice is an additive decomposition: 
\(
k((e,\bm x), (e',\bm x')) = k_E(e, e') + k_X(\bm x, \bm x'),
\)
where $k_E$ governs temporal smoothness and $k_X$ captures covariate similarity.
Because the two components enter additively, the temporal trend is shared across all covariate profiles. For $k_X$, the automatic relevance determination (ARD) squared exponential kernel \citep{Neal1996BayesianLearning} assigns a separate length-scale to each covariate dimension:
\(
k_{\rm ARD}(\bm x, \bm x') = \alpha^2 \exp\left\{-\sum_{d=1}^{p}(x_d - x_d')^2/(2h_d^{2})\right\},
\)
where $\alpha^2$ is the variance parameter and $h_1, \ldots, h_p$ are dimension-specific length-scales. We recommend standardizing each covariate to zero mean and unit variance before fitting, so that the length-scale parameters $h_d$ are on a comparable scale across dimensions. An alternative is a product kernel $k_E \times k_X$, which permits time-by-covariate interactions but can be problematic in moderate-to-high dimensions. The product form requires two points to be close in both enrollment time and all covariate dimensions for their covariance to be nonnegligible, so the covariance between distinct points vanishes rapidly as the dimension grows, effectively disabling information borrowing \citep{duvenaud_automatic_nodate}. The additive kernel avoids this pathology because similarity in any subset of dimensions, temporal or covariate, contributes positively to the overall covariance.

\section{Simulation}\label{sec: simulation}

We conduct simulation studies to assess the performance of the proposed single-task and multi-task GP methods in estimating treatment effects within a platform trial. The trial (Figure~\ref{fig: simulation diagram}) consists of two enrollment windows. During the first window (EW1), participants are randomized equally between Treatment 1 (control) and Treatment 2. In the second window (EW2), Treatment 3 is introduced and participants are randomized equally across all three arms. 

\begin{figure}[ht]
    \centering
    \includegraphics[width=0.6\linewidth]{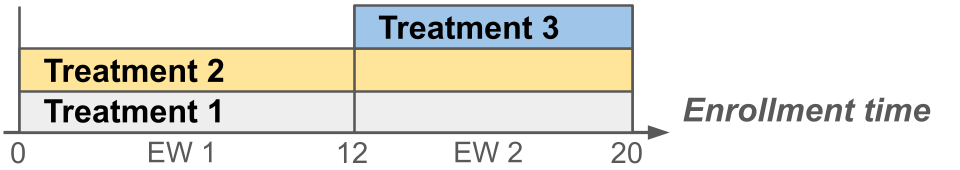}
    \caption{Platform trial design used in simulations.
        }
    \label{fig: simulation diagram}
\end{figure}

For each individual $i$, the enrollment time is $E_i \sim \textup{Unif}(0, 20)$, with $E_i \leq 12$ defining EW1 and $E_i > 12$ defining EW2. The primary estimand is the average treatment effect of Treatment~3 versus Treatment~1 over the ECE trial population, a superpopulation represented by individuals who enroll during EW2 when both treatments are simultaneously available. In the reference setting (Scenario 1), the data-generating process is $Y^{(1)}_i = 1 + \textup{expit}\left(E_i/2 - 6\right) + \epsilon_{i,1}$, $Y_i^{(2)}=\epsilon_{i,2}$, $Y^{(3)}_i = 5 + \left(E_i/10 - 1\right)^2 + \epsilon_{i,3}$, where $\epsilon_{i,1} \sim N(0, 4)$, $\epsilon_{i,2} \sim N(0, 1)$, $\epsilon_{i,3} \sim N(0, 1)$, and $\textup{expit}(z) = \{1 + \exp(-z)\}^{-1}$. The true value of $\delta_{31}$ is $3.58$. We also evaluate our method across four additional scenarios (Appendix~\ref{sec: additional scenarios}) that vary temporal trend smoothness, covariate effects, and outcome type.

We compare the proposed GP methods against several alternative approaches. We first describe methods that use concurrent data only, which serve as benchmarks for comparison, and then describe methods that can incorporate nonconcurrent data.
\begin{itemize}
  \item \textbf{ANOVA} (concurrent): mean difference $\bar{Y}_3 - \bar{Y}_1$ among concurrently enrolled participants.
  \item \textbf{ANHECOVA} (concurrent): a linear model adjusting for treatment-by-time interactions. Implemented via the \textsf{R} package \textsf{RobinCar} \citep{ye2021better, bannick2026robincar}.
  \item \textbf{LR-EW} (all): a linear model that includes a main effect for the enrollment window indicator to account for temporal differences \citep{Roig2025TreatmentControl}.
  \item \textbf{Bayesian Time Machine (TM)} (all): the Bayesian hierarchical model described in Section~\ref{sec: intro} \citep{saville_bayesian_2022}, implemented via the \textsf{R} package \textsf{NCC} with default settings \citep{NCC}.
\end{itemize}
To assess the impact of borrowing information from nonconcurrent controls, we also evaluate the GP method fitted using concurrent controls only, and compare it with the proposed GP model fitted using both concurrent and nonconcurrent controls.

All GP models use the squared exponential kernel. For the single-task model, the prior mean is a constant $m_a \sim N(0, \tau_a^2)$ with $\tau_a \sim N_+(1, 1)$, where the subscript $_+$ denotes truncation to positive values, the kernel variance parameter follows $\alpha^2 \sim N_+(1, 1)$, and the noise standard deviation $\sigma_a \sim N_+(0, 1)$. For the multi-task model, the noise standard deviations follow $\sigma_j, \sigma_k \sim N_+(0, 1)$, the kernel variance parameters follow $\alpha_f^2, \alpha_\Delta^2 \sim \textup{Exp}(1)$, and the constant treatment effect mean follows $\delta \sim N(0, 5^2)$. Across all GP models, the length-scale parameters follow data-adaptive inverse-gamma priors \citep{betancourt2020robustGP}, with shape and rate calibrated so that 98\% of the prior mass lies between a lower bound $\ell$ and an upper bound $u$, where $\ell = \max(1,\, q_{0.05})$, $u = 0.8 \times d_{\max}$, $q_{0.05}$ is the 5th percentile and $d_{\max}$ the maximum of all pairwise absolute differences among concurrent enrollment times.
Each simulation run includes $n = 200$ subjects (including those assigned to Treatment~2), and results are averaged across 1000 replicates. For Bayesian methods, MCMC sampling uses 2 chains of 2000 iterations each. Performance is evaluated by bias, standard deviation (SD), average standard error (SE), and empirical coverage probability (CP) at the 95\% level. For Bayesian methods, SE is the width of the 95\% credible interval divided by $2 \times 1.96$. For frequentist methods, SE is the square root of the estimated variance.

Figure \ref{fig:treatment_effect_multitask_GP} shows estimation results under one simulated dataset from the reference setting. The left panel displays the posterior mean functions of the control and treatment arms, with pointwise 95\% credible intervals, from the multi-task GP fitted using all data versus concurrent data only. The right panel shows the corresponding posterior mean and credible intervals for $\Delta$. Incorporating nonconcurrent controls narrows the credible intervals in both panels without introducing visible bias, reflecting successful borrowing of information across enrollment windows via temporal smoothness. A corresponding figure for the single-task GP is provided in the Supplementary Materials (Figure~\ref{fig: treatment effect single}) and shows a similar pattern.

\begin{figure}
    \centering
    \includegraphics[width=1\linewidth]{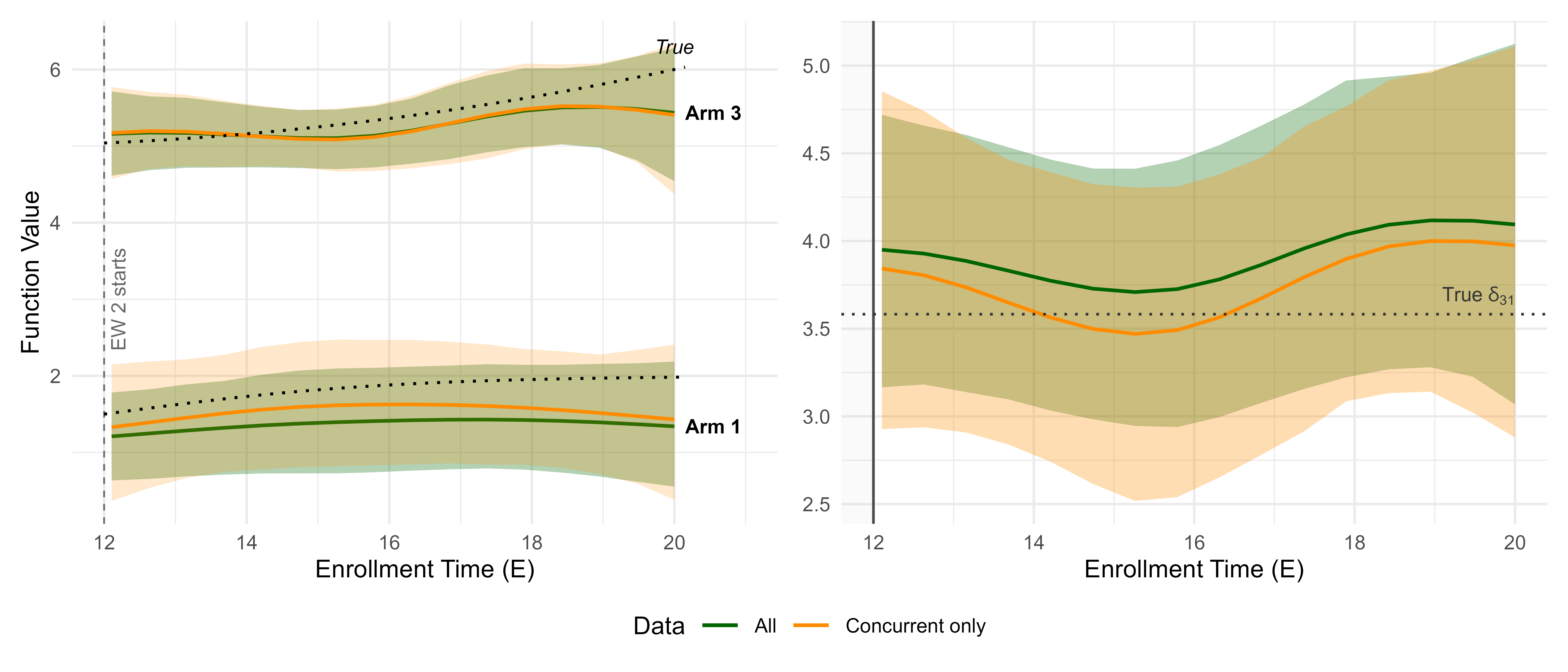}
    \caption{Estimation results under one simulated dataset from the reference setting. 
\textit{Left}: posterior mean functions of the control and treatment arms from the multi-task GP fitted using all data versus concurrent data only. \textit{Right}: posterior mean of $\Delta$. Shaded areas are pointwise 95\% credible intervals.}
\label{fig:treatment_effect_multitask_GP}
\end{figure} 

Table~\ref{tab:scenario1} presents results for Scenario~1 (reference). Results for Scenarios~2--5 are in Section~\ref{sec: additional scenarios} of the Supplementary Materials. In the reference setting, 
the concurrent-only benchmarks ANOVA and ANHECOVA have low bias and coverage near the nominal 95\% level (94.5\% and 93.1\%, respectively). The concurrent-only GP methods perform comparably. When incorporating nonconcurrent controls, the multi-task GP reduces SD from 0.429 to 0.357 relative to its concurrent-only counterpart, with a modest increase in bias (0.116 to 0.216) that leaves coverage near-nominal (94.4\%). Notably, the multi-task GP using all data achieves a smaller average standard error (0.395) than every concurrent-only method, including ANHECOVA (0.420), confirming an efficiency gain from incorporating nonconcurrent controls. It also achieves similar bias to the single-task GP and slightly lower variance.
Among methods that incorporate nonconcurrent controls, LR-EW exhibits bias (0.259) and below-nominal coverage (90.0\%), reflecting the inadequacy of adjusting for enrollment window alone to account for the nonlinear temporal trend. The Bayesian TM produces credible intervals an order of magnitude wider than those of every other method. This reflects known sensitivity of the Bayesian TM to its prior specification, particularly the priors governing the time-trend model \citep{Roig2025TreatmentControl}.

Results for the remaining scenarios (Section~\ref{sec: additional scenarios} of the Supplementary Materials) are broadly consistent. Overall, the GP-based estimators for incorporating nonconcurrent controls perform well across simulation settings. Under smooth temporal trends, they achieve small bias, near-nominal coverage, and meaningful efficiency gains relative to estimators that use only concurrent controls. Even in the most challenging setting with large and abrupt temporal shifts (Scenario 3), where the smooth GP prior is misspecified, the GP-based estimators maintain controlled bias and substantially outperform LR-EW (coverage 0.4\%) and Bayesian TM. Coverage does drop, to 90.2\% for the single-task GP and 82.8\% for the multi-task GP, reflecting the cost of working-model misspecification. We nonetheless recommend examining temporal trends, as they can signal elevated risk and less gain from borrowing nonconcurrent controls. In contrast, the competing methods for incorporating nonconcurrent controls exhibit less reliable behavior. The LR-EW estimator applies only a constant mean adjustment to nonconcurrent controls and incurs substantial bias when this assumption is violated. The Bayesian TM estimator shows unpredictable bias that can be quite large, and its standard deviation is markedly inflated relative to all other methods, including basic estimators that do not use nonconcurrent data at all. Its inference is also severely conservative, with coverage often reaching 100\%.

\begin{table}[ht]
\centering
\caption{Simulation results under Scenario 1 (reference setting). The true value of $\delta_{31}$ is  3.58.}
\label{tab:scenario1}
\resizebox{.7\linewidth}{!}{\begin{tabular}{llcccc}
\toprule
\textbf{Data} & \textbf{Method} & \textbf{Bias} & \textbf{SD} & \textbf{SE}  & \textbf{CP} \\
\midrule
\multirow{4}{*}{Concurrent only} 
& ANOVA            & 0.019 & 0.434 & 0.439 & 94.5 \\
& ANHECOVA         & 0.014 & 0.437 & 0.420 & 93.1 \\
& Single-task GP   & 0.092 & 0.433 & 0.441 & 94.8 \\
& Multi-task GP    & 0.116 & 0.429 & 0.434 & 94.6 \\
\midrule
\multirow{4}{*}{All data}
& Bayesian TM      & 0.101 & 0.741 & 4.895 & 100.0 \\
& LR-EW            & 0.259 & 0.335 & 0.357 & 90.0 \\
& Single-task GP   & 0.191 & 0.373 & 0.407 & 94.4 \\
& Multi-task GP    & 0.216 & 0.357 & 0.395 & 94.4 \\
\bottomrule
\end{tabular}}
\end{table}

\section{Real data application}\label{sec: real data}

We apply our GP methods to a hypothetical platform trial constructed from SURMOUNT-1~\citep{jastreboff2022tirzepatide}, a phase~3, double-blind, randomized controlled trial that enrolled 2539 adults with obesity (BMI $\geq 30$, or $\geq 27$ with at least one weight-related complication) across multiple countries. Participants were randomized 1:1:1:1 to once-weekly subcutaneous tirzepatide (TZP) at 5\,mg, 10\,mg, 15\,mg, or placebo for 72 weeks. The primary outcome was percent change in body weight from baseline to week~72.

To construct a hypothetical platform trial, we restrict to US sites and consider three arms: placebo, TZP~5\,mg, and TZP~10\,mg. Placebo and TZP~5\,mg begin enrolling from day~0, while TZP~10\,mg is added from day~170, corresponding to June~4, 2020, when enrollment resumed substantially after a slowdown due to COVID-19. We define the nonconcurrent control window as days~0--170, retaining only placebo and TZP~5\,mg participants, and the concurrent window as days~170+, retaining all three arms. This yields a total of 349 participants: 231 nonconcurrent placebo participants, 57 concurrent placebo participants, and 61 concurrent TZP~10\,mg participants. The estimand is the average treatment effect of TZP~10\,mg versus placebo in the ECE trial population. The 60 concurrent TZP~5\,mg participants are also concurrently eligible and contribute their enrollment times to the estimation, though their outcomes do not enter the analysis.

The primary outcome was missing for some participants, primarily due to treatment discontinuation. In the resulting hypothetical platform trial, 86 of 349 participants (24.6\%) had missing week-72 outcomes, with missingness highest among concurrent placebo participants (40.4\%), compared with 24.7\% for nonconcurrent placebo and 9.8\% for concurrent TZP~10\,mg. Following the pre-specified analysis of SURMOUNT-1~\citep{jastreboff2022tirzepatide}, we performed multiple imputation on the full trial. For COVID-19-related discontinuations ($n=16$), which are assumed missing at random, we fitted a mixed-model for repeated measures (MMRM) using longitudinal weight data (Weeks 4--72) from all subjects in the full trial excluding non-COVID treatment discontinuers to predict Week~72 outcomes for the 16 COVID dropouts. For non-COVID treatment discontinuations ($n=313$), which are handled under a retrieved-dropout framework, we applied predictive mean matching via \texttt{mice} separately within each treatment arm. 
This produced 50 multiply imputed complete datasets, from which the hypothetical platform trial was constructed. 
Results are pooled via Rubin's rules~\citep{Rubin1987MI}. We report both $\sqrt{W}$, the square root of the within-imputation variance (i.e., the average estimated variance across imputed datasets, reflecting sampling variability alone), and the total Rubin-pooled standard error.

For concurrent-only analyses, we apply ANOVA, ANHECOVA, and the single-task and multi-task GP models. For all-data analyses, we apply LR-EW, Bayesian TM, and the GP models. ANHECOVA adjusts for sex, presence of prediabetes, and baseline body weight. Results are reported in Table~\ref{tab:surmount1}. 

The concurrent-only methods agree closely. ANOVA gives $-18.6$ (SE~$2.42$), ANHECOVA $-18.8$ (SE~$2.36$), single-task GP $-18.6$ (SE~$2.45$) and multi-task GP $-18.8$ (SE~$2.41$), all consistent with the estimate in~\citet{jastreboff2022tirzepatide}. When incorporating nonconcurrent controls, the GP methods reduce within-imputation variance by 12\%--16\% relative to their concurrent-only counterparts while point estimates remain stable near $-18.8$, suggesting no meaningful bias from borrowing nonconcurrent controls. LR-EW assumes a constant shift between enrollment windows, yields a similar point estimate, and achieves $\sqrt{W} = 1.62$ (SE~$2.09$), the lowest within-imputation standard error among all methods. This strong performance likely reflects the absence of a meaningful temporal trend in this dataset (see Figure~\ref{fig: placebo temporal trend} in the Supplementary Materials). When a trend is present and a constant shift is insufficient to capture it, LR-EW can be severely biased, as shown in the simulations. 
The Bayesian TM yields a point estimate noticeably lower than those of the other methods and with more than doubled uncertainty, consistent with the prior sensitivity behavior discussed in Section~\ref{sec: simulation}. 

\begin{table}[htbp]
\centering
\caption{Estimated average treatment effect of TZP~10\,mg versus placebo on percent weight change at week~72, based on the hypothetical platform trial constructed from SURMOUNT-1. Results are pooled over $50$ imputations via Rubin's rules~\citep{Rubin1987MI}. $\sqrt{W}$: square root of within-imputation variance, reflecting sampling variability alone; SE: Rubin-pooled standard error, additionally accounting for missing-data uncertainty. All quantities are in percentage points.}
\label{tab:surmount1}
\resizebox{.8\linewidth}{!}{\begin{tabular}{@{}llcccc@{}}
\toprule
Data & Estimator & Estimate & $\sqrt{W}$ & SE & 95\% CI \\
\midrule
\multirow{4}{*}{Concurrent only ($n=118$)}
  & ANOVA        & $-18.6$ & $2.02$ & $2.42$ & $(-23.4,\ -13.9)$ \\
  & ANHECOVA                   & $-18.8$ & $1.94$ & $2.36$ & $(-23.4,\ -14.2)$ \\
  & Single-task GP             & $-18.6$ & $2.06$ & $2.45$ & $(-23.4,\ -13.8)$ \\
  & Multi-task GP              & $-18.8$ & $2.04$ & $2.41$ & $(-23.5,\ -14.1)$ \\
\midrule
\multirow{4}{*}{All data ($n=349$)}
  & LR-EW                      & $-18.6$ & $1.62$ & $2.09$ & $(-22.7,\ -14.5)$ \\
  & Bayesian TM                & $-11.6$ & $5.16$ & $5.19$ & $(-21.7,\ -1.4)$ \\
  & Single-task GP             & $-18.8$ & $1.89$ & $2.22$ & $(-23.1,\ -14.4)$ \\
  & Multi-task GP              & $-18.8$ & $1.91$ & $2.25$ & $(-23.2,\ -14.4)$ \\
\bottomrule
\end{tabular}}
\end{table}

\section{Discussion}

The GP models used throughout this paper require $O(n^3)$ computational cost for Cholesky decomposition, which may become a bottleneck in large platform trials such as the RECOVERY trial \citep{recovery2021dexamethasone}. However, scalable approximations are readily available. The Hilbert space approximate GP (HSGP) \citep{riutort2023practical} replaces the covariance matrix with a low-rank basis function expansion and can be applied directly to our single-task model. Using the \textsf{brms} package \citep{burkner2017brms} with $M=40$ basis functions, we repeated the 1000-replicate simulation under Scenario~1. The HSGP approximation produced nearly identical results to the exact single-task GP (bias 0.190 vs 0.191, SD 0.369 vs 0.373, 95\% CP 94.1\% vs 94.4\%) while reducing median computation time from 224 seconds to 6 seconds per replicate, a 37-fold speedup. Since HSGP scales as $O(nM)$ rather than $O(n^3)$, this approach extends the proposed method to large trials. The same basis function strategy can be applied to the multi-task GP model.

\bibliography{GP}

\clearpage

\appendix

\begin{center}
    {\LARGE\bf Supplementary Material for \\ 
``Robust and Data-Adaptive Integration of Nonconcurrent Data in Platform Trials via Gaussian Processes''} 
\end{center}

\setcounter{figure}{0}
\setcounter{equation}{0}
\setcounter{table}{0}
\setcounter{lemma}{0}
\setcounter{section}{0}
\renewcommand{\theequation}{S\arabic{equation}}
\renewcommand{\thetable}{S\arabic{table}}
\renewcommand{\thefigure}{S\arabic{figure}}
\renewcommand{\thesection}{S\arabic{section}}
\renewcommand{\thelemma}{S\arabic{lemma}}
\renewcommand{\thecondition}{S\arabic{condition}}

\section{Additional Figures}
\begin{figure}[!htbp]
    \centering
    \includegraphics[width=1\linewidth]{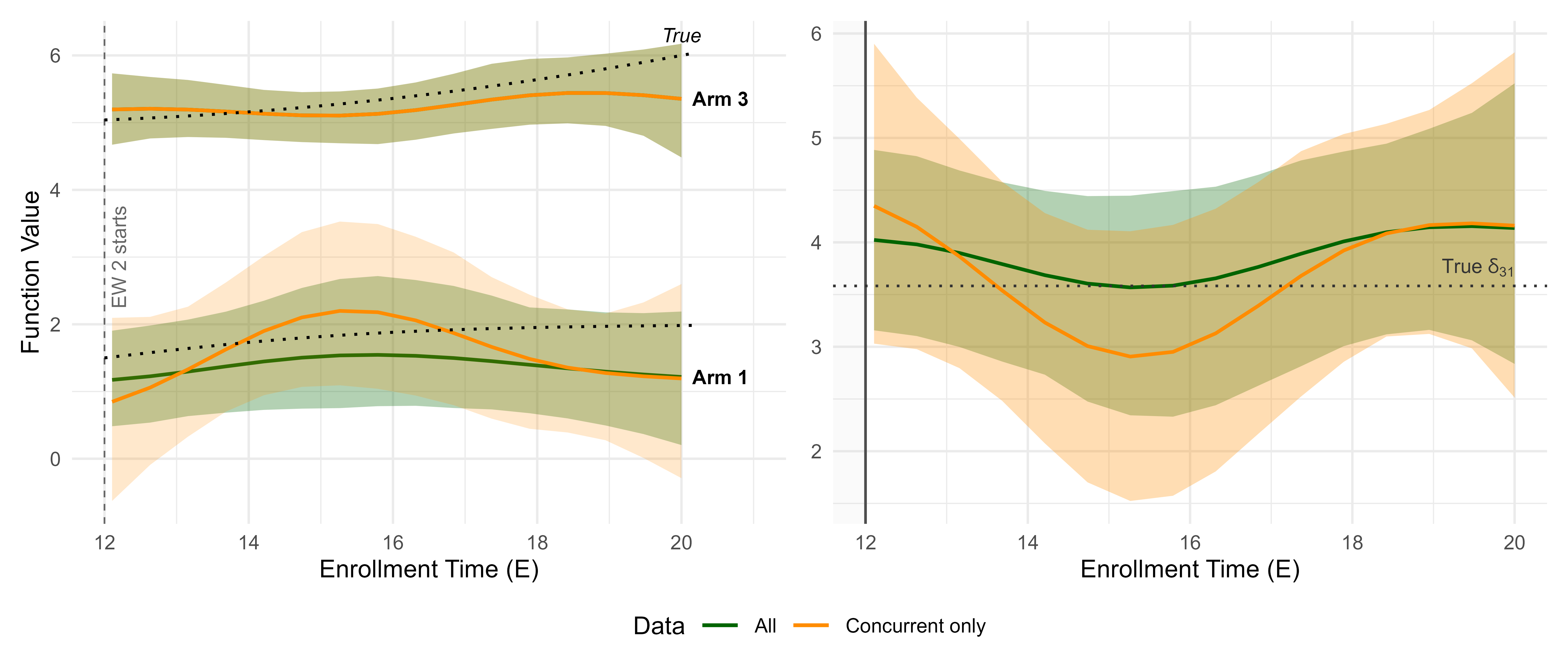}
    \caption{{Estimation results from the single-task GP under one simulated dataset from the reference setting, following the layout of Figure~\ref{fig:treatment_effect_multitask_GP}. As in the multi-task case, incorporating nonconcurrent controls narrows the credible intervals without visible bias, reflecting successful borrowing of information across enrollment windows via temporal smoothness.}}
    \label{fig: treatment effect single}
\end{figure}

\begin{figure}[!htbp]
    \centering
    \includegraphics[width=0.85\linewidth]{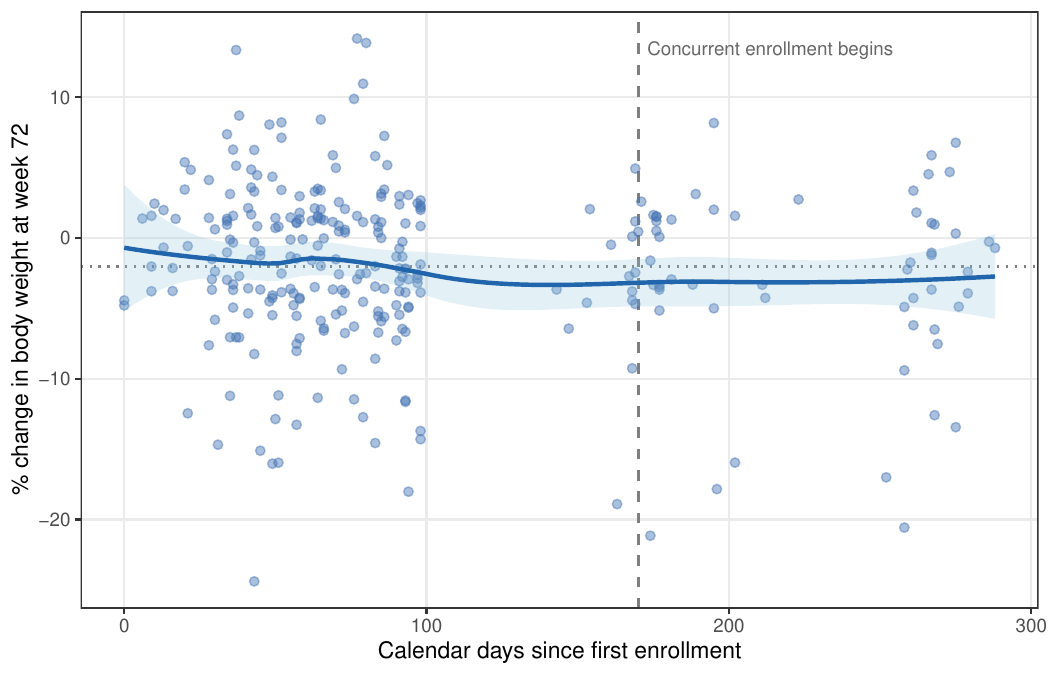}
    \caption{Placebo percent weight change at week~72 versus enrollment time (calendar days since first enrollment) in the hypothetical platform trial in Section \ref{sec: real data}. Each point represents a single placebo participant, with outcomes averaged across $50$ multiply imputed datasets. The solid curve is a LOESS smooth with a 95\% confidence band; the horizontal dotted line marks the overall placebo mean ($-2.0\%$). The dashed vertical line indicates when concurrent enrollment begins (day~170). The LOESS smooth remains approximately flat across the entire enrollment period, with the nonconcurrent mean ($-1.8\%$) and concurrent mean ($-2.8\%$) being nearly equal, reflecting the absence of a meaningful temporal trend in this dataset.}
    \label{fig: placebo temporal trend}
\end{figure}

\section{Additional Simulation Scenarios and Results}\label{sec: additional scenarios}
Scenarios~2--5 each modify some aspects of the reference setting (Scenario~1). The true value of $\delta_{31}$ is computed analytically where possible, or estimated via Monte Carlo with $n = 10^7$. Scenario~2 has $\delta_{31} = 4$; Scenario~3 has $\delta_{31} \approx 3.58$; Scenario~4 has $\delta_{31} \approx 3.38$; and Scenario~5 has $\delta_{31} \approx 0.12$. Note that Scenarios~2--4 have the low-noise specification $\epsilon_{i,1}\sim N(0,1)$ and $\epsilon_{i,3}\sim N(0,1/4)$.

\begin{enumerate}
\setcounter{enumi}{1}
    \item No temporal trend: $Y_i^{(1)}=1+\epsilon_{i,1}$ and $Y_i^{(3)}=5+\epsilon_{i,3}$.
    \item Nonsmooth temporal trend in Treatment~1: $Y_i^{(1)}=1 + \textup{expit}(E_i/2-6)-3\cdot\mathbb{I}(E_i<12)+\epsilon_{i,1}$; $Y_i^{(3)}=5+(E_i/10-1)^2+\epsilon_{i,3}$.
    \item Additional covariates: $Y_i^{(1)}=1+\textup{expit}(E_i/2-6) + X_{i,b}X_{i,c}/2 + \epsilon_{i,1}$, $Y_i^{(3)}=5+ (E_i/10-1)^2 +X_{i,c}/4 - X_{i,b}/2+ \epsilon_{i,3}$, where $X_{i,c} \sim N(0, 1/4)$, and $X_{i,b} \sim \textup{Bernoulli}(0.4)$. $X_{i,b}$ and $X_{i,c}$ are generated independently of $E_i$.
    \item Binary outcomes: $\mbb E(Y_i^{(1)}\mid E_i)=\textup{expit}(-1 + E_i/20)$ and $\mbb E(Y_i^{(3)}\mid E_i) =\textup{expit}(-0.5 + E_i/20)$.
\end{enumerate}

In Scenario~4, ANHECOVA adjusts only for treatment-by-time interactions and does not use the prognostic covariates $X_{i,b}$ and $X_{i,c}$, keeping the comparator identical across all scenarios. Within the GP framework, the comparison between the base GP and the linear-plus-kernel variant isolates the value of explicit covariate adjustment. For the GP models, covariates are standardized to zero mean and unit variance. The mean-function covariate coefficients follow $N(0, 1)$. The kernel is augmented with an additive squared exponential covariate kernel using automatic relevance determination, whose per-dimension length-scales follow data-adaptive inverse-gamma priors computed from all pairwise Euclidean distances among the standardized covariate vectors and whose variance parameter follows $\textup{Exp}(1)$. We evaluate two GP variants: the base model with a constant mean function and time-only kernel and a ``linear~+~kernel'' variant that combines the linear mean function with an additive ARD covariate kernel (Section~\ref{subsec: covariates}).

Tables~\ref{tab:scenario4}--\ref{tab:scenario8} present results for Scenarios~2--5. With no temporal trend (Scenario~2), all methods that incorporate nonconcurrent data (except for Bayesian TM) perform well; LR-EW achieves the lowest variance. Scenario~3 is the most challenging: the jump discontinuity at $E=12$ violates the smoothness assumption, and multi-task GP coverage drops to 82.8\% while single-task GP is more robust at 90.2\%; LR-EW coverage falls to 0.4\%. With additional covariates (Scenario~4), the linear-plus-kernel GP variants perform similarly to the unadjusted GP, while ANHECOVA, which uses only treatment-by-time adjustment, consistent with all other scenarios, provides an unadjusted-covariate benchmark. For binary outcomes (Scenario~5), the GP methods maintain near-nominal coverage with similar variance reductions.

\begin{table}[ht]
\centering
\caption{Simulation results under Scenario~2 (no temporal trend). The true value of $\delta_{31}$ is $4$.}
\label{tab:scenario4}
\begin{tabular}{llccccc}
\toprule
\textbf{Data} & \textbf{Method} & \textbf{Bias} & \textbf{SD} & \textbf{SE} & \textbf{CP} \\
\midrule
\multirow{4}{*}{Concurrent only}
& ANOVA & 0.009 & 0.216 & 0.217 & 94.0 \\
& ANHECOVA         & 0.009 & 0.218 & 0.210 & 92.6 \\
& Single-task GP   & 0.033 & 0.218 & 0.233 & 95.7 \\
& Multi-task GP    & 0.038 & 0.217 & 0.231 & 95.3 \\
\midrule
\multirow{4}{*}{All data}
& Bayesian TM      & -0.399 & 1.308 & 4.631 & 100.0 \\
& LR-EW            & 0.005 & 0.167 & 0.191 & 97.5 \\
& Single-task GP   & 0.011 & 0.172 & 0.197 & 97.5 \\
& Multi-task GP    & 0.017 & 0.174 & 0.200 & 97.4 \\
\bottomrule
\end{tabular}
\end{table}

\begin{table}[ht]
\centering
\caption{Simulation results under Scenario~3 (nonsmooth temporal trend). The true value of $\delta_{31}$ is $3.58$.}
\label{tab:scenario6}
\begin{tabular}{llccccc}
\toprule
\textbf{Data} & \textbf{Method} & \textbf{Bias} & \textbf{SD} & \textbf{SE} & \textbf{CP} \\
\midrule
\multirow{4}{*}{Concurrent only}
& ANOVA & 0.010 & 0.222 & 0.226 & 94.6 \\
& ANHECOVA         & 0.005 & 0.219 & 0.211 & 93.3 \\
& Single-task GP   & 0.037 & 0.222 & 0.239 & 95.4 \\
& Multi-task GP    & 0.039 & 0.219 & 0.234 & 95.7 \\
\midrule
\multirow{4}{*}{All data}
& Bayesian TM      & 0.187 & 0.946 & 4.829 & 100.0 \\
& LR-EW            & 1.291 & 0.268 & 0.279 & 0.4 \\
& Single-task GP   & 0.203 & 0.229 & 0.251 & 90.2 \\
& Multi-task GP    & 0.273 & 0.231 & 0.247 & 82.8 \\
\bottomrule
\end{tabular}
\end{table}

\begin{table}[ht]
\centering
\caption{Simulation results under Scenario~4 (additional covariates). The true value of $\delta_{31}$ is $3.38$.}
\label{tab:scenario7}
\begin{tabular}{llccccc}
\toprule
\textbf{Data} & \textbf{Method} & \textbf{Bias} & \textbf{SD} & \textbf{SE} & \textbf{CP} \\
\midrule
\multirow{6}{*}{Concurrent only}
& ANOVA & -0.011 & 0.243 & 0.236 & 94.4 \\
& ANHECOVA         & -0.013 & 0.238 & 0.219 & 91.5 \\
& Single-task GP   & 0.012 & 0.240 & 0.247 & 95.0 \\
& Single-task GP (linear + kernel) & 0.021 & 0.245 & 0.264 & 96.3 \\
& Multi-task GP    & 0.019 & 0.239 & 0.244 & 95.2 \\
& Multi-task GP (linear + kernel) & 0.025 & 0.240 & 0.249 & 95.1 \\
\midrule
\multirow{6}{*}{All data}
& Bayesian TM      & 0.202 & 1.242 & 4.681 & 100.0 \\
& LR-EW            & 0.243 & 0.194 & 0.202 & 78.2 \\
& Single-task GP   & 0.079 & 0.235 & 0.235 & 93.6 \\
& Single-task GP (linear + kernel) & 0.084 & 0.232 & 0.239 & 93.1 \\
& Multi-task GP    & 0.079 & 0.224 & 0.228 & 93.9 \\
& Multi-task GP (linear + kernel) & 0.084 & 0.222 & 0.230 & 93.5 \\
\bottomrule
\end{tabular}
\end{table}

\begin{table}[ht]
\centering
\caption{Simulation results under Scenario~5 (binary outcomes). The true value of $\delta_{31}$ is $0.12$.}
\label{tab:scenario8}
\begin{tabular}{llccccc}
\toprule
\textbf{Data} & \textbf{Method} & \textbf{Bias} & \textbf{SD} & \textbf{SE} & \textbf{CP} \\
\midrule
\multirow{4}{*}{Concurrent only}
& ANOVA & -0.001 & 0.138 & 0.138 & 94.2 \\
& ANHECOVA         & -0.003 & 0.140 & 0.133 & 92.9 \\
& Single-task GP   & -0.004 & 0.139 & 0.131 & 93.5 \\
& Multi-task GP    & -0.007 & 0.137 & 0.130 & 93.2 \\
\midrule
\multirow{4}{*}{All data}
& Bayesian TM      & 0.141 & 1.683 & 4.521 & 100.0 \\
& LR-EW            & 0.037 & 0.125 & 0.124 & 93.4 \\
& Single-task GP   & 0.004 & 0.131 & 0.127 & 94.0 \\
& Multi-task GP    & 0.002 & 0.129 & 0.125 & 94.0 \\
\bottomrule
\end{tabular}
\end{table}

\clearpage

\section{Proofs}
We repeatedly use the following two matrix facts.

\textbf{Fact 1.} If $B$ is positive semi-definite, then for any matrix $A$, $A^T BA$ is positive semi-definite.

\emph{Proof.} For any conformable vector $\bm x$, $(A\bm x)^T BA\bm x\geq0$.

\textbf{Fact 2.} If $B\succeq A\succ 0$, then $A^{-1}\succeq B^{-1}$.

\emph{Proof.} $B\succeq A\implies A^{-1/2}BA^{-1/2}\succeq \bm I\implies A^{-1}\succeq B^{-1}.$

\subsection{Proof of Theorem~\ref{theorem: variance reduction}(i)}
For simplicity, we suppress the subscript $a$ in the kernel $k_a$, posterior covariance function $\bar k_a$, and noise variance $\sigma_a^2$. Let $\bm E$ denote the vector of concurrent enrollment times for arm $a$ (the concurrent subvector of $\bm E_a$), and $\bm E'$ the vector of nonconcurrent enrollment times for arm $a$, of size $n'$. Then,
\begin{align*}
    &\quad \bar k_{\textup{cc}}(\bm e, \bm e)-\bar k(\bm e, \bm e)\\
    &=k(\bm e,(\bm E^T,(\bm E')^T)^T)\begin{pmatrix}
        k(\bm E,\bm E)+\sigma^2\bm I & k(\bm E,\bm E')\\
        k(\bm E',\bm E) & k(\bm E',\bm E')+\sigma^2\bm I
    \end{pmatrix}^{-1}
        k((\bm E^T,(\bm E')^T)^T,\bm e) -k(\bm e,\bm E)(k(\bm E,\bm E)+\sigma^2\bm I)^{-1}k(\bm E,\bm e)\\
    &=k(\bm e,(\bm E^T,(\bm E')^T)^T)\left\{{\underbrace{\begin{pmatrix}
        k(\bm E,\bm E)+\sigma^2\bm I & k(\bm E,\bm E')\\
        k(\bm E',\bm E) & k(\bm E',\bm E')+\sigma^2\bm I
    \end{pmatrix}}_{M}}^{-1}-\underbrace{\begin{pmatrix}
        (k(\bm E,\bm E)+\sigma^2\bm I)^{-1} & \mathbf{0}\\
       \mathbf{0} & \mathbf{0}
    \end{pmatrix}}_N\right\}k((\bm E^T,(\bm E')^T)^T,\bm e)
\end{align*}
By the Schur complement formula and the invertibility of $M_{11}$, we have
\begin{align*}
    &\quad M^{-1}-N\\
    &=\begin{pmatrix}
        M_{11}^{-1}M_{12}S_{11}^{-1}M_{21}M_{11}^{-1} & -M_{11}^{-1}M_{12}S_{11}^{-1}\\
        -S_{11}^{-1}M_{21}M_{11}^{-1} & S_{11}^{-1}
    \end{pmatrix}\\
    &=\begin{pmatrix}
        -M_{11}^{-1}M_{12}\\
        \bm I_{n'}
    \end{pmatrix}S_{11}^{-1}\begin{pmatrix}
        -M_{21}M_{11}^{-1} &
        \bm I_{n'}
    \end{pmatrix}\\
    &=\begin{pmatrix}
        -M_{11}^{-1}M_{12}\\
        \bm I_{n'}
    \end{pmatrix}S_{11}^{-1}\begin{pmatrix}
        -M_{11}^{-1}M_{12}\\
        \bm I_{n'}
    \end{pmatrix}^T
\end{align*}
where $S_{11}=M/M_{11}=M_{22}-M_{21}M_{11}^{-1}M_{12}$ is the Schur complement, which is positive definite because $M$ is a kernel Gram matrix plus $\sigma^2 \bm I$ and hence positive definite.

Therefore, $M^{-1}-N$ is positive semi-definite, and by Fact~1, $\bar k_{\textup{cc}}(\bm e, \bm e) - \bar k(\bm e, \bm e)$ is positive semi-definite.

\subsection{Proof of Theorem~\ref{theorem: variance reduction}(ii)}
Some calculations show that
\(
    \Delta(\bm e)\mid \bm Y,\bm E\sim N(\bar m_\Delta(\bm e),\bar {k}_\Delta(\bm e, \bm e)),
\)
with 
\begin{align*}
    \bar m_\Delta(\bm e)&=\delta + k_{\Delta}(\bm e, \bm{E}_j) B_1^{-1} \left\{(\bm{Y}_j - m(\bm{E}_j)-\delta\bm 1_{n_j}) + B_2 (\bm{Y}_{k} - m(\bm{E}_k))\right\}\\
    \bar{k}_\Delta(\bm e, \bm e)&=k_{\Delta}(\bm e,\bm e)-k_{\Delta}(\bm e,\bm E_j)B_1^{-1} k_{\Delta}(\bm E_j,\bm e),
\end{align*}
where 
\begin{align*}
    B_1&= k_{f,jj}+k_{\Delta, jj}+\sigma_j^2\bm I_{n_j}-k_{f,jk}(k_{f,kk}+\sigma_{k}^2\bm I_{n_k})^{-1} k_{f,kj}\\
    B_2&=-k_{f,jk}(k_{f,kk}+\sigma_k^2\bm I_{n_k})^{-1}.
\end{align*}
Suppose we incorporate nonconcurrent controls $\bm E_{k'}$ with $\bm E_k$ to construct $(\bm E_{k},\bm E_{k'})$ (i.e., do not include extra concurrent controls or any treatment data). The efficiency gain is
\begin{align*}
    \bar{k}_{\Delta,\textup{cc}}(\bm e, \bm e)-\bar{k}_\Delta(\bm e, \bm e)=k_{\Delta}(\bm e,\bm E_j)\left(B_1^{-1}-B_{1,\textup{cc}}^{-1}\right)k_{\Delta}(\bm E_j,\bm e).
\end{align*}
To show the last matrix is positive semi-definite, it suffices to show $B_{1,\textup{cc}}-B_1$ is positive semi-definite. Note that $B_1$ is the Schur complement of $k_{f,kk}+\sigma_k^2\bm I_{n_k}$ in the positive definite joint covariance matrix of $(\bm{Y}_j^T,\bm{Y}_k^T)^T$. Hence $B_1\succ 0$ in both the concurrent-only and full-data settings, and Fact~2 applies. Let the vector of control enrollment times used in the full-data setting be $(\bm E_k,\bm E_{k'})$.
\begin{align*}
    &\quad B_{1,\textup{cc}}-B_1\\
    &=-k_{f,jk}(k_{f,kk}+\sigma_{k}^2\bm I_{n_k})^{-1}k_{f,kj}+\begin{pmatrix}
        k_{f,kj}\\
        k_{f,k'j}
    \end{pmatrix}^T\begin{pmatrix}
        k_{f,kk} + \sigma_k^2\bm I_{n_k} & k_{f,kk'}\\
        k_{f,k'k} & k_{f,k'k'}+\sigma_k^2\bm I_{n_{k'}}
    \end{pmatrix}^{-1}\begin{pmatrix}
        k_{f,kj}\\
        k_{f,k'j}
    \end{pmatrix}\\
    &=\begin{pmatrix}
        k_{f,kj}\\
        k_{f,k'j}
    \end{pmatrix}^T\left[{\underbrace{\begin{pmatrix}
        k_{f,kk} + \sigma_k^2\bm I_{n_k} & k_{f,kk'}\\
        k_{f,k'k} & k_{f,k'k'}+\sigma_k^2\bm I_{n_{k'}}
    \end{pmatrix}}_{M}}^{-1}-\underbrace{\begin{pmatrix}
        (k_{f,kk} + \sigma_k^2\bm I_{n_k})^{-1} & 0\\
     0  & 0
    \end{pmatrix}}_{N}\right]\begin{pmatrix}
        k_{f,kj}\\
        k_{f,k'j}
    \end{pmatrix}
\end{align*}
By the Schur complement formula and the invertibility of $M_{11}$, we have that
\begin{align*}
    &\quad M^{-1}-N\\
    &=\begin{pmatrix}
        M_{11}^{-1}M_{12}S_{11}^{-1}M_{21}M_{11}^{-1} & -M_{11}^{-1}M_{12}S_{11}^{-1}\\
        -S_{11}^{-1}M_{21}M_{11}^{-1} & S_{11}^{-1}
    \end{pmatrix}\\
    &=\begin{pmatrix}
        -M_{11}^{-1}M_{12}\\
        \bm I_{n_{k'}}
    \end{pmatrix}S_{11}^{-1}\begin{pmatrix}
        -M_{21}M_{11}^{-1} &
        \bm I_{n_{k'}}
    \end{pmatrix}\\
    &=\begin{pmatrix}
        -M_{11}^{-1}M_{12}\\
        \bm I_{n_{k'}}
    \end{pmatrix}S_{11}^{-1}\begin{pmatrix}
        -M_{11}^{-1}M_{12}\\
        \bm I_{n_{k'}}
    \end{pmatrix}^T
\end{align*}
where $S_{11}=M/M_{11}=M_{22}-M_{21}M_{11}^{-1}M_{12}$ is the Schur complement, which is positive definite because $M=k_f((\bm{E}_k,\bm{E}_{k'}),(\bm{E}_k,\bm{E}_{k'}))+\sigma_k^2 \bm I_{n_k+n_{k'}}$ is positive definite.

Therefore, $M^{-1}-N$ is positive semi-definite. By Fact~1, $B_{1,\textup{cc}} - B_1$ is positive semi-definite, completing the proof.

\subsection{Proof of Theorem~\ref{theorem: bias bound}}
Throughout, all statements are conditional on $(\bm E_j,\bm E_k,(E_i)_{i\in\mathcal{I}_{jk}})$, and all kernel hyperparameters are treated as fixed. Let $\bm Y = (\bm Y_j^T, \bm Y_k^T)^T$ and $\bm E = (\bm E_j^T, \bm E_k^T)^T$ denote the stacked outcome and enrollment-time vectors. We use $F_{\textup{Sch}} := k_{f,jj} - k_{f,jk}(k_{f,kk}+\sigma_k^2 \bm I_{n_k})^{-1} k_{f,kj}$ for the Schur complement that appears in the decomposition $B_1 = k_{\Delta,jj} + \sigma_j^2 \bm I_{n_j} + F_{\textup{Sch}}$.

Recall that $\hat\psi = B^{-1}\sum_{b=1}^B \hat\delta_{jk}^{(b)}$, where
$\hat\delta_{jk}^{(b)} = \sum_{i\in\mathcal{I}_{jk}} p_i^{(b)}\,
\hat\Delta^{(b)}(E_i)$, the $E_i$ are the observed enrollment times of the
$n_{jk}$ concurrently eligible individuals,
$(p_i^{(b)})_{i\in\mathcal{I}_{jk}} \sim \textup{Dir}(1,\dots,1)$, and
$\hat\Delta^{(b)}(\cdot)$ are posterior draws of $\Delta(\cdot)$. Denote the
observed data $\{\bm Y,\bm E,(E_i)_{i\in\mathcal{I}_{jk}}\}$ by $\mathcal{D}$. By construction, $\mathbb{E}[p_i^{(b)}] = 1/n_{jk}$, the posterior draws are conditionally independent and identically distributed given $\mathcal{D}$, and the Dirichlet weights are independent of the posterior draws given $\mathcal{D}$. Therefore
\[
\mathbb E[\hat\psi\mid \mathcal{D}]
= \frac{1}{n_{jk}}\sum_{i\in\mathcal{I}_{jk}}
\bar m_\Delta(E_i)
= \mathbb{E}_{F_n}[\bar m_\Delta(E)].
\]
Substituting the posterior mean
from~\eqref{eq: multi-GP posterior mean and covariance} gives
\[
\mathbb E[\hat\psi\mid \mathcal{D}]
= \delta + \mu_n(\bm E_j) B_1^{-1}
\left\{(\bm Y_j-m(\bm E_j)-\delta \bm 1_{n_j})
+ B_2(\bm Y_k-m(\bm E_k))\right\},
\]
where $\mu_n(\bm E_j)
= n_{jk}^{-1}\sum_{i\in\mathcal{I}_{jk}} k_\Delta(E_i, \bm E_j)
\in \mathbb{R}^{1\times n_j}$.
By the consistency assumption and Assumption~\ref{assump:dgp},
\[
\mathbb E[\bm Y_j\mid \bm E_j]=f_0(\bm E_j)+\Delta_0(\bm E_j),\qquad
\mathbb E[\bm Y_k\mid \bm E_k]=f_0(\bm E_k).
\]
Taking the expectation over $\bm Y$ given $(\bm E_j, \bm E_k)$ by the property of total expectation, and subtracting the sample-level estimand
$\tilde\delta_{jk} = \mathbb{E}_{F_n}[\Delta_0(E)]$, the estimation bias is
\begin{align*}
&\mathbb{E}[\hat\psi \mid \bm E_j, \bm E_k, (E_i)_{i\in\mathcal{I}_{jk}}] - \tilde\delta_{jk} \\
&\quad= \delta + \mu_n(\bm E_j)B_1^{-1}
\bigl\{(\Delta_0(\bm E_j)-\delta \bm 1_{n_j})
+ (f_0(\bm E_j)-m(\bm E_j))
+ B_2(f_0(\bm E_k)-m(\bm E_k))\bigr\}
- \tilde\delta_{jk}.
\end{align*}
With $\bm h := \Delta_0(\bm E_j) - \delta\, \bm 1_{n_j}$ and
$\bm g := f_0(\bm E_j)-m(\bm E_j)+B_2(f_0(\bm E_k)-m(\bm E_k))$, we
decompose the estimation bias as
\[
\mathbb{E}[\hat\psi \mid \bm E_j, \bm E_k, (E_i)_{i\in\mathcal{I}_{jk}}] - \tilde\delta_{jk}
= \underbrace{(\delta - \tilde\delta_{jk})
  + \mu_n(\bm E_j)\, B_1^{-1}\, \bm h}_{
  A:\;\text{treatment effect}}
+ \underbrace{\mu_n(\bm E_j)\, B_1^{-1}\, \bm g}_{
  B:\;\text{baseline}}.
\]
We note the following special cases:
\begin{enumerate}
    \item If the prior mean of $f$ is correctly specified, i.e., $m=f_0$, then the baseline deviation vanishes:
    \[\mathbb{E}[\hat\psi \mid \bm E_j, \bm E_k, (E_i)_{i\in\mathcal{I}_{jk}}] - \tilde\delta_{jk}
    = (\delta - \tilde\delta_{jk})
    +\mu_n(\bm E_j)B_1^{-1}\bm h.\]
    \item If the prior mean of $f$ is correctly specified and the true underlying treatment effect function is time-invariant, i.e., $\Delta_0\equiv\delta_{jk}$, then
    $$\mathbb{E}[\hat\psi \mid \bm E_j, \bm E_k, (E_i)_{i\in\mathcal{I}_{jk}}] - \tilde\delta_{jk}
    =(\delta-\delta_{jk})(1-\mu_n(\bm E_j)B_1^{-1}\bm 1_{n_j}).$$
    \item If additionally $\delta=\delta_{jk}$, then
    $$\mathbb{E}[\hat\psi \mid \bm E_j, \bm E_k, (E_i)_{i\in\mathcal{I}_{jk}}] -\tilde\delta_{jk}= 0.$$
\end{enumerate}

We now establish the general upper bound.

\textbf{Treatment effect term.} Define $h = \Delta_0 - \delta$, so that
$\bm h = h(\bm E_j)$, and the pointwise interpolation error
\[
b_h(e) = \sum_{i=1}^{n_j} w_i(e)\, h(E_{j,i}) - h(e), \qquad
w_i(e) = [B_1^{-1} k_\Delta(\bm E_j, e)]_i.
\]
We verify that $A = \mathbb{E}_{F_n}[b_h(E)]$. Indeed, writing
$b_h(e) = k_\Delta(e,\bm E_j)\,B_1^{-1}\bm h - h(e)$ and averaging
under $F_n$,
\begin{align*}
\mathbb{E}_{F_n}[b_h(E)]
&= \mu_n(\bm E_j)\,B_1^{-1}\bm h
   - \mathbb{E}_{F_n}[\Delta_0(E) - \delta] \\
&= \mu_n(\bm E_j)\,B_1^{-1}\bm h
   + \delta - \tilde\delta_{jk}
= A.
\end{align*}
By the reproducing property of $\mathcal{H}_{k_\Delta}$
(Lemma~3.9 of \citealp{kanagawa_gaussian_2018}),
\[
    |b_h(e)| \leq \|h\|_{\mathcal{H}_{k_\Delta}}\,
    \bigl\| k_\Delta(\cdot, e)
    - \textstyle\sum_i w_i(e)\,
    k_\Delta(\cdot, E_{j,i})
    \bigr\|_{\mathcal{H}_{k_\Delta}}.
\]
Expanding the RKHS norm squared with
$\bm v(e) = k_\Delta(\bm E_j, e)\in\mathbb R^{n_j}$ gives
\begin{align*}
    \bigl\| k_\Delta(\cdot, e)
    - \sum_i w_i(e)\, k_\Delta(\cdot, E_{j,i})
    \bigr\|^2_{\mathcal{H}_{k_\Delta}}
    &= k_\Delta(e,e)
       - 2\,k_\Delta(e,\bm E_j)\,B_1^{-1}\,k_\Delta(\bm E_j,e) \\
    &\quad+ k_\Delta(e,\bm E_j)\,B_1^{-1}\,k_{\Delta,jj}\,
       B_1^{-1}\,k_\Delta(\bm E_j,e)\\
    &= \bar{k}_\Delta(e,e)
       - \bm v(e)^T B_1^{-1}
       \bigl(B_1-k_{\Delta,jj}\bigr) B_1^{-1}\, \bm v(e)\\
    &= \bar{k}_\Delta(e,e)
       - \bm v(e)^T B_1^{-1}
       \bigl(\sigma_j^2 \bm I_{n_j} + F_{\textup{Sch}}\bigr)
       B_1^{-1}\, \bm v(e)\\
    &\leq \bar{k}_\Delta(e,e),
\end{align*}
where we used
$B_1 - k_{\Delta,jj} = \sigma_j^2 \bm I_{n_j} + F_{\textup{Sch}} \succeq 0$.
Therefore, $b_h(e)^2 \leq \|h\|_{\mathcal{H}_{k_\Delta}}^2\,
\bar{k}_\Delta(e,e)$ for every $e$. Since $A = \mathbb{E}_{F_n}[b_h(E)]$,
Jensen's inequality gives
\[
    A^2 \leq \mathbb{E}_{F_n}[b_h(E)^2]
    \leq \|\Delta_0 - \delta\|_{\mathcal{H}_{k_\Delta}}^2\,
    \mathbb{E}_{F_n}[\bar{k}_\Delta(E,E)].
\]

\textbf{Baseline term.} By the reproducing property of
$\mathcal{H}_{k_f}$, for each $i$,
$g_i = \langle f_0 - m,\, r_i \rangle_{\mathcal{H}_{k_f}}$,
where
\[
    r_i(\cdot) = k_f(\cdot, E_{j,i})
    - \sum_{l=1}^{n_k} c_{il}\, k_f(\cdot, E_{k,l}), \quad
    c_{il} = \bigl[k_f(\bm E_j, \bm E_k)
    (k_{f,kk} + \sigma_k^2 \bm I_{n_k})^{-1}\bigr]_{il}.
\]
Letting $\bm a = B_1^{-1}\, \mu_n(\bm E_j)^T \in \mathbb{R}^{n_j}$, the
baseline component becomes
\[
    B = \mu_n(\bm E_j)\, B_1^{-1}\, \bm g
    = \sum_i a_i\, g_i
    = \bigl\langle f_0 - m,\;
    \textstyle\sum_i a_i\, r_i \bigr\rangle_{\mathcal{H}_{k_f}}.
\]
By the Cauchy--Schwarz inequality in $\mathcal{H}_{k_f}$,
\[
    B^2 \leq \|f_0 - m\|_{\mathcal{H}_{k_f}}^2 \cdot \bm a^T G_f\, \bm a,
\]
where $[G_f]_{ii'} = \langle r_i, r_{i'} \rangle_{\mathcal{H}_{k_f}}$.
Expanding via the reproducing property yields, for any $i,i'$
\begin{align*}
[G_f]_{ii'}
&=\left\langle
k_f(\cdot,E_{j,i})-\sum_{l}c_{il}k_f(\cdot,E_{k,l}),
k_f(\cdot,E_{j,i'})-\sum_{l'}c_{i'l'}k_f(\cdot,E_{k,l'})
\right\rangle_{\mathcal H_{k_f}} \\
&=
\left\langle k_f(\cdot,E_{j,i}),k_f(\cdot,E_{j,i'})\right\rangle
-\sum_{l'} c_{i'l'} \left\langle k_f(\cdot,E_{j,i}),k_f(\cdot,E_{k,l'})\right\rangle \\
&\quad
-\sum_l c_{il}\left\langle k_f(\cdot,E_{k,l}),k_f(\cdot,E_{j,i'})\right\rangle
+\sum_{l,l'} c_{il}c_{i'l'}\left\langle k_f(\cdot,E_{k,l}),k_f(\cdot,E_{k,l'})\right\rangle\\
&=k_f(E_{j,i},E_{j,i'})
-\sum_{l'} c_{i'l'}k_f(E_{j,i},E_{k,l'})-\sum_l c_{il}k_f(E_{k,l},E_{j,i'})
+\sum_{l,l'} c_{il}c_{i'l'}k_f(E_{k,l},E_{k,l'}).
\end{align*}
Collecting these entries across all $i,i'$, we get
\begin{equation*}
    G_f=k_{f,jj}-2k_{f,jk}(k_{f,kk}+\sigma_k^2\bm I_{n_k})^{-1}k_{f,kj}+k_{f,jk}(k_{f,kk}+\sigma_k^2\bm I_{n_k})^{-1}k_{f,kk}(k_{f,kk}+\sigma_k^2\bm I_{n_k})^{-1}k_{f,kj}.
\end{equation*}
Letting
$P_k := (k_{f,kk} + \sigma_k^2 \bm I_{n_k})^{-1}$ gives
\begin{align*}
    G_f
    &= k_{f,jj} - 2\,k_{f,jk}\, P_k\, k_{f,kj}
       + k_{f,jk}\, P_k\, k_{f,kk}\, P_k\, k_{f,kj} \\
    &= F_{\textup{Sch}} - \sigma_k^2\, k_{f,jk}\, P_k^2\, k_{f,kj},
\end{align*}
using $\bm I_{n_k} - k_{f,kk}\, P_k = \sigma_k^2\, P_k$. 

As a Gram matrix,
$G_f \succeq 0$. Since
$\sigma_k^2\, k_{f,jk}\, P_k^2\, k_{f,kj} \succeq 0$, we have
$G_f \preceq F_{\textup{Sch}}$.

Combining both terms via
$|\mathbb{E}[\hat\psi \mid \bm E_j, \bm E_k, (E_i)_{i\in\mathcal{I}_{jk}}] - \tilde\delta_{jk}|
\leq |A| + |B|$ with
$\bm a^T G_f\, \bm a
= \mu_n(\bm E_j)\, B_1^{-1}\, G_f\, B_1^{-1}\, \mu_n(\bm E_j)^T$
yields~\eqref{eq: bias bound thm}.

\subsection{Proof of Corollary~\ref{cor:noninflating_bias_bound}}
Let $\mu := \mu_n(\bm E_j)^T \in \mathbb{R}^{n_j}$ and
$Q = \mu^T B_1^{-1} G_f B_1^{-1} \mu$. Since $G_f\succeq 0$, the matrix
$\mc M := B_1^{-1} G_f B_1^{-1}$ is positive semidefinite. Noting that
$\mu = \mathbb{E}_{F_n}[\bm v(E)]$ where
$\bm v(e) = k_\Delta(\bm E_j, e) \in \mathbb{R}^{n_j}$, and that
$x \mapsto x^T \mc M x$ is convex, Jensen's inequality under $F_n$ gives
\[
Q = (\mathbb E_{F_n}[\bm v(E)])^T \mc M (\mathbb E_{F_n}[\bm v(E)])
\leq \mathbb{E}_{F_n}\left[\bm v(E)^T \mc M \bm v(E)\right]
\leq \mathbb E_{F_n}\left[\bm v(E)^T B_1^{-1} F_{\textup{Sch}} B_1^{-1}
\bm v(E)\right],
\]
where the second inequality uses $G_f \preceq F_{\textup{Sch}}$. From the
proof of the treatment effect bound,
$\bar{k}_\Delta(e,e) - \bm v(e)^T B_1^{-1}
(\sigma_j^2 \bm I_{n_j} + F_{\textup{Sch}}) B_1^{-1} \bm v(e) \geq 0$ for
every $e$, so in particular
\[
\bm v(e)^T B_1^{-1} F_{\textup{Sch}} B_1^{-1} \bm v(e)
\leq \bar k_\Delta(e,e).
\]
Therefore,
$Q \leq \mathbb E_{F_n}[\bar k_\Delta(E,E)]$, which gives
\[
|B| \leq \|f_0-m\|_{\mathcal{H}_{k_f}}\,
\left[\mathbb E_{F_n}\{\bar k_\Delta(E,E)\}\right]^{1/2}.
\]
Combining with the treatment effect bound
$|A| \leq \|\Delta_0 - \delta\|_{\mathcal{H}_{k_\Delta}}\,
[\mathbb{E}_{F_n}\{\bar{k}_\Delta(E,E)\}]^{1/2}$ via the triangle inequality
yields~\eqref{eq: bias bound cor}. The right-hand side does not increase when
nonconcurrent controls are incorporated because
$\mathbb{E}_{F_n}\{\bar{k}_\Delta(E,E)\}
= n_{jk}^{-1}\sum_{i\in\mathcal{I}_{jk}} \bar{k}_\Delta(E_i,E_i)$ and
each $\bar{k}_\Delta(E_i,E_i)$ does not increase by
Theorem~\ref{theorem: variance reduction}(ii).
\end{document}